\documentclass[pdflatex,sn-basic]{sn-jnl}%

\usepackage{graphicx}%
\usepackage{multirow}%
\usepackage{amsmath,amssymb,amsfonts}%
\usepackage{amsthm}%
\usepackage{mathrsfs}%
\usepackage[title]{appendix}%
\usepackage{xcolor}%
\usepackage{textcomp}%
\usepackage{manyfoot}%
\usepackage{booktabs}%
\usepackage{algorithm}%
\usepackage{algorithmicx}%
\usepackage{algpseudocode}%
\usepackage{listings}%
\usepackage{tikz}
\usepackage{subcaption}
\usepackage{amsthm}
\usepackage{thmtools}
\usepackage{amsmath,amsfonts,amsthm,amssymb,multirow,commath}
\usepackage{algorithm}
\usepackage{algorithmicx}
\usepackage{algpseudocode}
\usepackage{thm-restate}
\usepackage{etoolbox}
\usepackage{enumitem,mathtools,bbm}
\usepackage{footmisc,dsfont}
\usepackage{hyperref}
\usepackage{cleveref}
\usepackage{setspace}
\usepackage{graphicx}
\usepackage{placeins}
\usepackage{xcolor}
\usepackage{dsfont}
\usepackage{units}

\definecolor{DustyRose}{HTML}{DC6B82}

\theoremstyle{thmstyleone}%
\newtheorem{theorem}{Theorem}%
\newtheorem{proposition}[theorem]{Proposition}%
\newtheorem{lemma}[theorem]{Lemma}
\newtheorem{corollary}[theorem]{Corollary}

\theoremstyle{thmstyletwo}%
\newtheorem{example}{Example}%
\newtheorem{assumption}{Assumption}
\newtheorem{conjecture}{Conjecture}

\theoremstyle{thmstylethree}%
\newtheorem{definition}{Definition}%

\begin{document}

\title[Nash Core in Multiwinner Election]{Nash Core in Multiwinner Election}

\author*[1]{\fnm{Ashish} \sur{Goel}}\email{ashishg@stanford.edu}

\author*[1]{\fnm{Zhihao} \sur{Jiang}}\email{zhihao@alumni.stanford.edu}

\author*[2]{\fnm{Chenghan} \sur{Zhou}}\email{chzhou@stanford.edu}

\affil*[1]{\orgdiv{Department of Management Science and Engineering}, \orgname{Stanford University}, \orgaddress{\street{475 Via Ortega}, \city{Stanford}, \postcode{94305}, \state{California}, \country{USA}}}

\affil*[2]{\orgdiv{Department of Computer Science}, \orgname{Stanford University}, \orgaddress{\street{353 Jane Stanford Way}, \city{Stanford}, \postcode{94305}, \state{California}, \country{USA}}}

\abstract{In the approval-based committee selection problem, a committee is said to be in the core if no subset of voters has an incentive to deviate by selecting a \emph{blocking} committee of proportional size, such that every voter in the deviating group strictly prefers the blocking committee.

We consider the setting where candidates can be selected fractionally. Under a mild regularity assumption, we show that there always exists a weighting of candidates such that the fractional committee maximizing the candidate-weighted Nash Social Welfare is in the core. We refer to such a solution as being in the \emph{Nash core}. Additionally, we show that a Nash core solution admits a payment assignment between voters and candidates, where each voter pays a candidate they approve in proportion to the weight.

For the discrete setting, where each candidate is either included or excluded from the committee, we prove that every approval-based committee election with at most eight equally weighted voters has a core committee by rounding the fractional Nash core solution. Although the non-emptiness of the core in this setting remains an open question and checking core membership is coNP-hard, we extend the notion of the Nash core to the discrete case, yielding a formulation that is efficiently verifiable and offers a promising path toward establishing core existence in discrete settings.

Finally, we test our approach on real voting data using a payment-guided heuristic. We empirically show that the Nash core solution can be efficiently computed through an iterative algorithm in both the fractional and discrete settings.

}


\maketitle

\renewcommand{\P}[1]{\mathbf{Pr}\left(#1\right)}
\newcommand{\Px}[2]{\mathbf{Pr}_{#1}\left(#2\right)}
\newcommand{\E}[1]{\mathbb{E}\left[#1\right]}
\newcommand{\Ex}[2]{\mathbb{E}_{#1}\left[#2\right]}
\newcommand{\I}[1]{\mathbb{I}\left(#1\right)}

\newcommand{\ot}{\tilde{O}}

\newcommand{\alert}[1]{\textcolor{red}{#1}}

\newcommand{\mA}{\mathcal{A}}
\newcommand{\vecone}{\mathbbm{1}}
\newcommand{\onenorm}[1]{\|#1\|_{1}}
\newcommand{\cx}[1]{x(#1)}
\newcommand{\cxs}[1]{x^{*}(#1)}
\newcommand{\cy}[1]{y(#1)}
\newcommand{\real}{\mathbb{R}}
\newcommand{\intg}{\mathbb{N}}
\newcommand{\pZ}{(P0)}
\newcommand{\pa}{(P1)}
\newcommand{\pb}{(P2)}
\newcommand{\pc}{(P3)}
\renewcommand{\pd}{(P4)}
\newcommand{\pe}{(P5)}
\newcommand{\zerovec}{\mathbf{0}}
\newcommand{\unitvec}{\mathbf{e}}
\newcommand{\block}[1]{Bl\left(#1\right)}
\newcommand{\setw}{\mathcal{W}}

\section{Introduction}

Fairness in multi-winner approval voting has received growing attention, driven by the desire to ensure equitable representation for all voter groups. Among the most prominent fairness concepts are \emph{proportionality}, \emph{justified representation (JR)}, and the \emph{core}. Proportionality requires that cohesive groups of like-minded voters---those approving the same set of candidates---obtain representation roughly commensurate with their size in the electorate, thereby preventing minority groups from being entirely overshadowed by the majority \citep{aziz2017justified,faliszewski2017multiwinner}.

Building on this idea, JR strengthens the guarantee against disenfranchisement: it demands that any sufficiently large group must have at least one representative they collectively approve. This ensures that no significant coalition of voters is left entirely without direct representation in the elected committee \citep{aziz2017justified}. Meanwhile, the core focuses on \emph{stability}: a committee is in the core if no coalition of voters can unilaterally choose a different committee that all members of that coalition strictly prefer. Despite being conceptually well-defined and deeply studied, whether a core committee always exists in unrestricted approval voting elections remains open \citep{faliszewski2017multiwinner,lackner2023multi}. We resolve this question for elections with at most eight individual equally weighted voters while leaving the unrestricted problem open.

\subsection{Model and the Core}

In the approval-based committee selection problem, we are given a set of voters $N = [n] = \{1, 2, \dots, n\}$ and a set of candidates $C = [m]$. Each voter $v \in N$ submits a ballot $A_v \subseteq C$. We assume that $A_v\neq\emptyset$ for every voter $v\in N$. If a candidate $c \in C$ is in $A_v$, we say that $v$ \emph{approves} $c$. A committee is a subset of candidates $W\subseteq C$. Given two committees $W_1$ and $W_2$, we say a voter $v$ strictly prefers $W_1$ to $W_2$, denoted by $W_1\succ_v W_2$, if $v$ approves more candidates in $W_1$ than $W_2$.

We study both the fractional and discrete settings. In the fractional setting, candidates can be selected in fractional amounts, and we use $x\in [0,1]^m$ to denote a committee. In the discrete setting, a committee is denoted by $x\in \{0,1\}^m$. In both cases, the size of a committee is $\|x\|_1 = \sum_{c=1}^{m}x_c$. Under this notation, voter $v$ strictly prefers $x$ to $y$ if $\sum_{c\in A_v}x_c>\sum_{c\in A_v}y_c$.

Our objective is to select a committee of size $k$ that reflects the collective preferences of $n$ voters. A desired property of the committee is proportionality, which is a key principle in committee selection, tracing its origins back more than a century \citep{droop1881methods}. In recent years, various formalizations of this concept have been extensively studied \citep{aziz2017justified,aziz2018complexity,brams2007minimax,chamberlin1983representative,fain2016core,fain2018fair,monroe1995fully,sanchez2017proportional}. The central idea is that any group of voters should feel sufficiently represented, removing any motivation to break away and form a more favorable, smaller committee. When the precise demographic coalitions are not known in advance, a robust solution concept is typically employed, requiring the committee to remain stable against any potential voter subset that might deviate.

We study the core, an extensively studied notion of fairness \citep{foley1970lindahl,lindahl1958just,muench1972core,samuelson2024pure,scarf1967core}. We define the core for both the fractional and discrete settings.

\begin{definition}[fractional core] \label{def:fraccore}

A committee $x\in [0,1]^m$ with $\|x\|_1\leq k$ is in the core, if for any nonzero $y\in [0,1]^m$,
\begin{align*}
    \left| \left\{ v\in N: \sum_{c\in A_v}y_c>\sum_{c\in A_v}x_c  \right\} \right|< \frac{n}{k}\|y\|_1.
\end{align*}

\end{definition}

\begin{definition}[discrete core \citep{aziz2017justified}] \label{def:disc:def}

A committee $x\in \{0,1\}^m$ with $\|x\|_1\leq k$ is in the discrete core if, for every nonzero $y\in \{0,1\}^m$,
\begin{align*}
    \left| \left\{ v\in N: \sum_{c\in A_v}y_c>\sum_{c\in A_v}x_c  \right\} \right| < \frac{n}{k}\|y\|_1.
\end{align*}

\end{definition}

In the discrete setting, $k$ is required to be an integer.

These definitions admit a taxation interpretation \citep{foley1970lindahl,scarf1967core,muench1972core}. Each voter is assigned one dollar. When the goal is to form a committee of size $k$, each candidate is priced at $\frac{n}{k}$. A committee $x$ is deemed stable if there is no group of voters who strictly prefer another committee $y$ and can collectively purchase $y$. In other words, the number of voters in this group must be strictly less than the cost of $y$, which is $\frac{n}{k}\|y\|_1$.

The fractional core is always non-empty \citep{foley1970lindahl,fain2016core} and can be computed in polynomial time \citep{kroer2025computing}. In the discrete setting, unrestricted core nonemptiness remains open, and determining membership in the discrete core is coNP-hard \citep{brill2024approval}. Section~\ref{sec:bv-core-eight} resolves the existence question for elections with at most eight individual equally weighted voters.

\subsection{Candidate-Weighted Nash Social Welfare} \label{sec:intro:fraccore}

In this subsection, we only consider the fractional setting.

We say the utility of voter $v$ given committee $x$ is $u_v(x):=\sum_{c\in A_v}x_c$. Nash social welfare is defined as the product of voter utilities, or equivalently, the sum of log-utilities $\sum_{v\in N}\log\left(u_v(x)\right)$. When each candidate can be selected in any quantity---even exceeding one---or equivalently, when each candidate has infinite number of copies, the committee that maximizes the Nash social welfare is in the core \citep{fain2016core}. However, this approach fails when each candidate can be chosen at most one unit. This fact is noted in several works \citep{peters2020proportionality,peters2025core,lackner2023multi}. The following counterexample illustrates this issue.

\begin{example} \label{exp:nsw}

Consider the fraction case with $n=3$ voters, $m=5$ candidates, and $k=3$ committee members. The figures below shows voters' approval sets. $v_1$ approves $\{c_1,c_3\}$, $v_2$ approves $\{c_1,c_4\}$, and $v_3$ approves $c_2,c_5$. The committee $x=(1,1,\frac{1}{3},\frac{1}{3},\frac{1}{3})$ (see Figure \ref{fig:nsw:exp:a}) maximizes Nash social welfare, but it is not in the core, because $y=(1,0,\frac{1}{2},\frac{1}{2},0)$ is preferred by both $v_1$ and $v_2$, and they can afford $y$. A core solution is in Figure \ref{fig:nsw:exp:b}, which is $x=(1,1,\frac{1}{2},\frac{1}{2},0)$.

\definecolor{cmt}{rgb}{0.8,0.9,1.0}

\begin{figure}[htbp]
\centering
\tikzset{global scale/.style={
    scale=#1,
    every node/.append style={scale=#1}
  }
}

\begin{subfigure}[t]{0.32\textwidth}
\centering
\begin{tikzpicture}%
\fill[cmt] (0,0) rectangle (3,0.8);
\node at (0.5,-0.3) {$v_1$};
\node at (1.5,-0.3) {$v_2$};
\node at (2.5,-0.3) {$v_3$};
\draw (0,0) rectangle (2,0.6);
\node at (1,0.3) {$c_1$};
\draw (2,0) rectangle (3,0.6);
\node at (2.5,0.3) {$c_2$};
\draw (0,0.6) rectangle (1,1.2);
\node at (0.5,0.9) {$c_3$};
\draw (1,0.6) rectangle (2,1.2);
\node at (1.5,0.9) {$c_4$};
\draw (2,0.6) rectangle (3,1.2);
\node at (2.5,0.9) {$c_5$};
\end{tikzpicture}
\caption{max Nash social welfare}
\label{fig:nsw:exp:a}
\end{subfigure}
\begin{subfigure}[t]{0.32\textwidth}
\centering
\begin{tikzpicture}%
\fill[cmt] (0,0) rectangle (2,0.9);
\fill[cmt] (2,0) rectangle (3,0.6);
\node at (0.5,-0.3) {$v_1$};
\node at (1.5,-0.3) {$v_2$};
\node at (2.5,-0.3) {$v_3$};
\draw (0,0) rectangle (2,0.6);
\node at (1,0.3) {$c_1$};
\draw (2,0) rectangle (3,0.6);
\node at (2.5,0.3) {$c_2$};
\draw (0,0.6) rectangle (1,1.2);
\node at (0.5,0.9) {$c_3$};
\draw (1,0.6) rectangle (2,1.2);
\node at (1.5,0.9) {$c_4$};
\draw (2,0.6) rectangle (3,1.2);
\node at (2.5,0.9) {$c_5$};
\end{tikzpicture}
\caption{committee in the core}
\label{fig:nsw:exp:b}
\end{subfigure}
\begin{subfigure}[t]{0.32\textwidth}
\centering
\begin{tikzpicture}%
\fill[cmt] (0,0) rectangle (2,0.6);
\fill[cmt] (2,0) rectangle (3,0.6);
\node at (0.5,-0.3) {$v_1$};
\node at (1.5,-0.3) {$v_2$};
\node at (2.5,-0.3) {$v_3$};
\draw (0,0) rectangle (2,0.3);
\node at (1,0.15) {$c_1$};
\draw (2,0) rectangle (3,0.6);
\node at (2.5,0.3) {$c_2$};
\draw (0,0.3) rectangle (1,0.9);
\node at (0.5,0.6) {$c_3$};
\draw (1,0.3) rectangle (2,0.9);
\node at (1.5,0.6) {$c_4$};
\draw (2,0.6) rectangle (3,1.2);
\node at (2.5,0.9) {$c_5$};
\end{tikzpicture}
\caption{\centering max weighted Nash SW\\(height corresponds to weight)}
\label{fig:nsw:exp:c}
\end{subfigure}
\caption{Committee selection for Example \ref{exp:nsw}}
\end{figure}
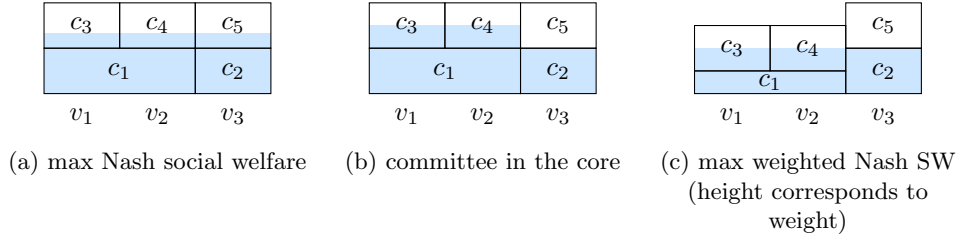

\end{example}

We introduce a weight $w_c$ to each candidate $c$. Given a committee $x$, define the weighted utility of voter $v$ as $\phi_v(x):=\sum_{c\in A_v}w_cx_c$, and the candidate-weighted Nash social welfare is $\sum_{v\in N}\log\left( \phi_v(x) \right)$.

In the above Example \ref{exp:nsw}, define the weight vector as $w=(\frac{1}{2},1,1,1,1)$, and re-draw the core committee in Figure \ref{fig:nsw:exp:b} as Figure \ref{fig:nsw:exp:c}, illustrating the weight change of $c_1$. This committee maximizes the candidate-weighted Nash social welfare.

We interpret the role of the weight \(w\) as follows. When a candidate \(c\) is popular---approved by many voters---selecting \(c\) raises the utilities of those voters. Because the Nash social welfare is the product of all voter utilities, it can then be more beneficial to select candidates supported by different voters, rather than adding additional candidates approved by the same group already approving \(c\).

For instance, in Figure~\ref{fig:nsw:exp:a}, \(c_1\) is popular, so in addition to \(c_1\) the committee mostly includes candidates approved by voter \(v_3\). However, since voters \(v_1\) and \(v_2\) together deserve two seats, the resulting committee violates core stability.

By introducing the weight \(w\) and assigning a smaller weight to the popular candidate \(c\), the weighted Nash social welfare solution becomes more inclined to select additional candidates approved by those who already support \(c\). This removes the effective penalty for choosing popular candidates and thus yields a fairer committee overall.

Based on this candidate-weighted idea, we formally define the notion of \emph{Nash core} as following.

\begin{definition}[fractional Nash core] \label{def:nashcore}
    For $w\in (0, 1]^m$ and $x\in [0,1]^m$ where $\|x\|_1\leq k$, we say $(w,x)$ is in the fractional Nash core, if both of the following hold:
    \begin{itemize}
        \item (Weight-selection coupling) For each candidate $c$, $x_c=1$ if $w_c<1$.
        \item (Nash optimality) $z=x$ is an optimal solution of maximizing the candidate-weighted Nash social welfare $\sum_{v\in N}\log\left( \sum_{c\in A_v}w_cz_c \right)$\footnote{We use the convention $\log 0=-\infty$.} subject to $\|z\|_1\leq k$ and $z\geq 0$.
    \end{itemize}
\end{definition}

We call it Nash ``core'', although we do not require it to lie in the core. However, Theorem \ref{thm:frac:incore} proves that if $(w,x)$ is in the Nash core, then $x$ is in the core, thereby justifying the name.

We note $z\leq 1$ is not required, meaning $x$ is optimal even comparing with those $z$ with $z_c>1$ for some $c$.

\subsection{Proportional Payment} \label{sec:intro:proppay}

In this subsection, we still only consider the fractional setting, but study it in the perspective of payment rules. It is discussed in the discrete setting in \citep{peters2020proportionality,peters2021market,sanchez2017proportional}. As mentioned in \citep{peters2021market}, sometimes there are multiple committees in the core, and some of them may not be ``fair enough''. We illustrate this fact using the following example.

\begin{example} \label{exp:multicore}

Consider the fraction case with $n=2$ voters, $m=3$ candidates, and $k=2$ committee members. The figures below shows voters' approval sets. $v_1$ approves $\{c_1,c_2\}$, and $v_2$ approves $\{c_1,c_3\}$. The three committees in Figure \ref{fig:multicore:exp:a}-\ref{fig:multicore:exp:c} are all in the core. However, since $c_2$ and $c_3$ are symmetric, we believe the committee $x=(1,\frac{1}{2},\frac{1}{2})$ in Figure \ref{fig:multicore:exp:c} is more fair.

\definecolor{cmt}{rgb}{0.8,0.9,1.0}

\begin{figure}[htbp]
\centering
\tikzset{global scale/.style={
    scale=#1,
    every node/.append style={scale=#1}
  }
}

\begin{subfigure}[t]{0.20\textwidth}
\centering
\begin{tikzpicture}%
\fill[cmt] (0,0) rectangle (2,0.6);
\fill[cmt] (0,0.6) rectangle (1,1.2);
\node at (0.5,-0.3) {$v_1$};
\node at (1.5,-0.3) {$v_2$};
\draw (0,0) rectangle (2,0.6);
\node at (1,0.3) {$c_1$};
\draw (0,0.6) rectangle (1,1.2);
\node at (0.5,0.9) {$c_2$};
\draw (1,0.6) rectangle (2,1.2);
\node at (1.5,0.9) {$c_3$};
\end{tikzpicture}
\caption{core 1}
\label{fig:multicore:exp:a}
\end{subfigure}
\begin{subfigure}[t]{0.20\textwidth}
\centering
\begin{tikzpicture}%
\fill[cmt] (0,0) rectangle (2,0.6);
\fill[cmt] (1,0.6) rectangle (2,1.2);
\node at (0.5,-0.3) {$v_1$};
\node at (1.5,-0.3) {$v_2$};
\draw (0,0) rectangle (2,0.6);
\node at (1,0.3) {$c_1$};
\draw (0,0.6) rectangle (1,1.2);
\node at (0.5,0.9) {$c_2$};
\draw (1,0.6) rectangle (2,1.2);
\node at (1.5,0.9) {$c_3$};
\end{tikzpicture}
\caption{core 2}
\label{fig:multicore:exp:b}
\end{subfigure}
\begin{subfigure}[t]{0.20\textwidth}
\centering
\begin{tikzpicture}%
\fill[cmt] (0,0) rectangle (2,0.6);
\fill[cmt] (0,0.6) rectangle (2,0.9);
\node at (0.5,-0.3) {$v_1$};
\node at (1.5,-0.3) {$v_2$};
\draw (0,0) rectangle (2,0.6);
\node at (1,0.3) {$c_1$};
\draw (0,0.6) rectangle (1,1.2);
\node at (0.5,0.9) {$c_2$};
\draw (1,0.6) rectangle (2,1.2);
\node at (1.5,0.9) {$c_3$};
\end{tikzpicture}
\caption{core 3}
\label{fig:multicore:exp:c}
\end{subfigure}
\begin{subfigure}[t]{0.30\textwidth}
\centering
\begin{tikzpicture}%
\fill[cmt] (0,0) rectangle (2,0.6);
\node at (0.5,-0.3) {$v_1$};
\node at (1.5,-0.3) {$v_2$};
\draw (0,0) rectangle (2,0.3);
\node at (1,0.15) {$c_1$};
\draw (0,0.3) rectangle (1,0.9);
\node at (0.5,0.6) {$c_2$};
\draw (1,0.3) rectangle (2,0.9);
\node at (1.5,0.6) {$c_3$};
\end{tikzpicture}
\caption{\centering redraw core 3\\(height corresponds to payment)}
\label{fig:multicore:exp:d}
\end{subfigure}
\caption{Committee selection for Example \ref{exp:multicore}}
\end{figure}
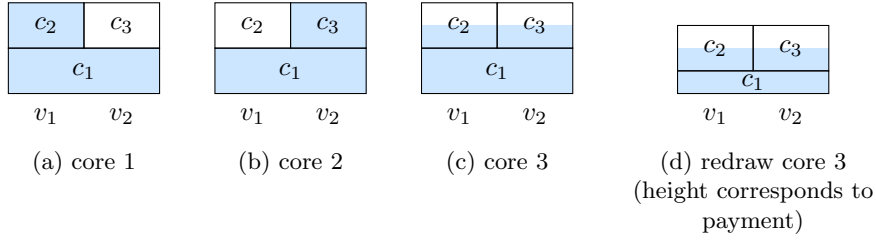

\end{example}

Recall the taxation interpretation we introduced after defining the fractional core in Definition~\ref{def:fraccore}: each voter has one dollar, and the price of each candidate is \(\frac{n}{k}\). In Example~\ref{exp:multicore}, the price of each candidate is also set to one dollar. We interpret this by examining how the voters pay for the committee members. For the first core committee \(x=(1,1,0)\) shown in Figure~\ref{fig:multicore:exp:a}, voter \(v_1\) pays for \(c_2\), which is approved exclusively by \(v_1\), while voter \(v_2\) pays for \(c_1\), which is approved by both \(v_1\) and \(v_2\). Observe that \(v_1\) does not pay for \(c_1\) yet still gains utility from it; this illustrates how the committee's unfairness is also reflected in the payment structure.

\medskip

To address this issue, we propose a \emph{proportional payment} approach. Under this scheme, each voter spends her (one) dollar in proportion to the weighted selected amounts $w_cx_c$ on the committee members she approves. However, if many voters approve a particular candidate, the total amount they collectively spend on that candidate could exceed \(\frac{n}{k}\). To avoid such an overpayment, we introduce a \emph{discount} \(w_c \in (0,1]\) for each candidate \(c\). Concretely, we say a payment is a \emph{proportional payment} if each voter allocates her budget among the committee members she approves in proportion to the weighted selected amounts $w_cx_c$. Specifically, when \(x_c < 1\), we set \(w_c = 1\), meaning there is no discount and thus no adjustment is needed. On the other hand, if \(x_c = 1\) (i.e., candidate \(c\) is fully selected), we may reduce \(w_c\) below \(1\) to ensure that the total payment for \(c\) equals the nominal price of \(\frac{n}{k}\). This discount \(w_c\) effectively lowers the collective payment on the popular candidate and helps maintain fairness in the resulting committee.

We formally define the proportional payment as following.
\begin{definition}[proportional payment]    \label{def:payment}
    For $w\in (0, 1]^m$ and $x\in [0,1]^m$ where $\|x\|_1\leq k$, we define the weighted utility of voter $v$ as $\phi_v(x)=\sum_{c\in A_v}w_cx_c$. If $\phi_v(x)>0$ for all $v\in N$, we also define a \emph{payment function} $p_{vc}=\frac{w_cx_c}{\phi_v(x)}$ if $c\in A_v$, and $p_{vc}=0$ if $c\notin A_v$. We say $(w,x)$ is supported by a \emph{proportional payment}, if all of the following holds:
    \begin{itemize}
        \item (Positive utility) For each $v\in N$, $\phi_v(x)>0$.
        \item (Affordability) For any $c\in C$, $\sum_{v\in N}p_{vc}=\frac{n}{k}x_c$.
        \item (Stability) For any $c\in C$, if $x_c<1$, then $\sum_{v:c\in A_v}\frac{1}{\phi_v(x)}\leq \frac{n}{k}$.
        \item (Weight-selection coupling) $w_c<1$ only when $x_c=1$.
    \end{itemize}
\end{definition}

In this definition, the last condition ensures that the discount applies only to candidates that are fully selected. Note that this condition is already implied by affordability and stability when $0 < x_c < 1$, and we still include it in order to ensure $w_c = 1$ for $x_c = 0$.

Now, we motivate the first three conditions. For each voter $v$, the payment $p_{vc}$ is proportional to how much $x_c$ contributes to $\phi_v(x)$, which is the reason why we call it proportional payment. Note that for each voter $v\in N$, $\sum_{c\in C}p_{vc}=1$, corresponding to each voter has one dollar. The second condition, affordability, corresponds to the price of each candidate is $n/k$. As for the third condition, stability, if a candidate $c$ violates this condition, it implies that candidate $c$ has ``enough'' support to be fully bought. If $\sum_{v:c\in A_v} \frac{1}{\phi_v(x)} > \frac{n}{k}$, then the total ``willingness to pay'' of $c$’s supporters exceeds the price $n/k$ of $c$. That suggests $x_c$ should be raised if possible, contradicting the assumption that $x_c < 1$. The stability condition rules out such scenarios, ensuring that if a candidate is not fully selected, it is precisely because its supporters do not, in aggregate, have sufficient incentive to cover the cost.

This proportional payment corresponds to a separable personalized price in the Lindahl equilibrium, and we will discuss more details in Section \ref{sec:lindahl}.

It might be observed that both the fractional Nash core in Definition \ref{def:nashcore} and the proportional payment in Definition \ref{def:payment} are defined for pair $(w,x)$. We will show they are equivalent definitions (Theorem \ref{thm:frac:equiv}).

\subsection{Our Results}

We summarize our results in this subsection, including results mentioned in previous subsections.

For the fractional Nash-core existence result, we impose the following assumption.

{
\begin{assumption}\label{assup:1}
    For any non-empty voter set $\emptyset\subsetneq S\subseteq N$, $|\bigcup_{v\in S}A_v|>\frac{k|S|}{n}$.
\end{assumption}
}

Assumption~\ref{assup:1} excludes coalitions whose approval union fits within their proportional committee share. It is needed only for fractional Nash-core existence in Section~\ref{sec:frac:exist}. For fractional-core existence, one can inductively reduce along a violating coalition to a smaller instance. The bounded-voter theorem in Section~\ref{sec:bv-core-eight} uses the same reduction.

In Section \ref{sec:intro:fraccore}, we introduced the fractional Nash core, and in Section \ref{sec:intro:proppay}, we discussed the concept of proportional payment. As a reminder, We motivated the fractional Nash core via Nash social welfare and motivated proportional payment via separable market price. We now show that these two concepts are equivalent.

\begin{restatable}{theorem}{fracequiv}
\label{thm:frac:equiv}
    For $w\in (0, 1]^m$ and $x\in [0,1]^m$ where $\|x\|_1\leq k$, $(w,x)$ is in the fractional Nash core if and only if $(w,x)$ is supported by a proportional payment.
\end{restatable}

We prove this theorem by considering the KKT condition of the optimization problem in the definition of the fractional Nash core. The proof is deferred to Appendix \ref{sec:proof:frac:equiv}.

We also show fractional Nash core solution is in the fractional core in Section \ref{sec:frac:incore}.

\begin{restatable}{theorem}{fracincore}
\label{thm:frac:incore}
    For $w\in (0, 1]^m$ and $x\in [0,1]^m$ where $\|x\|_1\leq k$, if $(w,x)$ is in the fractional Nash core, then $x$ is in the fractional core.
\end{restatable}

Most importantly, we prove that the fractional Nash core is non-empty under Assumption~\ref{assup:1} in Section \ref{sec:frac:exist}.

\begin{restatable}{theorem}{fracexist}
\label{thm:frac:exist}
    Under Assumption~\ref{assup:1}, there exists $w\in (0, 1]^m$ and $x\in [0,1]^m$ where $\|x\|_1\leq k$, such that $(w,x)$ is in the fractional Nash core.
\end{restatable}

We prove this theorem by formulating a concave game \citep{rosen1965existence}, in which a pure Nash equilibrium exists, and we show this pure Nash equilibrium corresponds to a fractional Nash core solution.

From a computational perspective, \citep{kroer2025computing} introduces a convex program whose optimal solution lies in the fractional core. We show that every fractional Nash core solution corresponds to an optimal solution of this program.

\begin{theorem}[informal]\label{thm:compute}
    Every fractional Nash core solution is an optimal solution of the convex program introduced in \citep{kroer2025computing}.
\end{theorem}

We prove this theorem in Section \ref{sec:compute}.

Our main exact result in the discrete setting is that every approval-based committee election with at most eight equally weighted voters has a committee in the discrete core (Section~\ref{sec:bv-core-eight}). The parameter is the number of individual voters, not the number of distinct approval types with arbitrary multiplicities. The proof combines the fractional Nash-core structure with exact finite verification after reducing the unresolved problem to a binary residual matrix with at most eight columns.

As a complementary approach for unrestricted elections, we define the discrete Nash core and a weaker variant called the weak Nash core in Section~\ref{sec:disc}. The discrete Nash core selects a committee that maximizes the sum of candidate-weighted harmonic utilities. In contrast, the weak Nash core requires that no unselected candidate would significantly increase this utility if added to the committee. Furthermore, checking whether $(w,x)$ belongs to the weak Nash core can be done efficiently. We now present the following property.

\begin{restatable}{theorem}{discincore}
\label{thm:disc:incore}
    Suppose $w\in (0,1]^m$ and $x\in \{0,1\}^m$ where $\|x\|_1\leq k$. If $(w,x)$ is in the discrete Nash core, then it is in the weak Nash core. If $(w,x)$ is in the weak Nash core, then $x$ is in the discrete core.
\end{restatable}

We complement these theoretical results with experiments on real voting data. Although our iterative heuristics lack formal convergence guarantees, they efficiently find candidate fractional Nash-core pairs and weak Nash-core pairs on the tested instances. Rather than verifying discrete-core membership directly---a coNP-hard problem \citep{brill2024approval}---we certify that the computed solution lies in the weak Nash core, which can be checked in polynomial time. To the best of our knowledge, this is the first algorithmic approach to find discrete-core solutions on real-world approval-voting data.

Based on the experiments, we conjecture that there always exists a weak Nash core solution.

\begin{conjecture}
There exists $w\in (0,1]^m$ and $x\in \{0,1\}^m$ where $\|x\|_1\leq k$, such that $(w,x)$ is in the weak Nash core.
\end{conjecture}

If this conjecture is true, it would imply that the discrete core is nonempty for unrestricted elections, suggesting a path beyond the bounded-voter result proved in Section~\ref{sec:bv-core-eight}.

\subsection{Related Work}

\subsubsection{The Lindahl Equilibrium}

Lindahl equilibrium, first introduced by \citep{lindahl1958just}, describes an efficient allocation mechanism for public goods where individuals pay personalized prices based on their marginal benefit. Samuelson \citep{samuelson2024pure} formalized the conditions for optimal public good provision but highlighted strategic challenges, such as the incentive to misrepresent preferences. Later, Arrow \citep{arrow1969organization} and Hurwicz \citep{hurwicz1979outcome} explored mechanism design issues, leading to the development of incentive-compatible frameworks.

A Lindahl equilibrium describes a way to fund and decide on the level of a public good by assigning each individual a ``personalized price'' that matches their valuation of the good, so that the total of these prices covers its cost. In this scenario, everyone contributes according to how much they benefit, ensuring no one pays more or less than their true valuation and leading to a Pareto-efficient outcome. Under our proportional payment framework, the personalized price remains separable as the ratio of candidate weight to voter utility.

\subsubsection{Justified Representation}

The discrete core implies several fairness concepts, including justified representation (JR), extended justified representation (EJR) \citep{aziz2017justified}, and proportional justified representation (PJR) \citep{sanchez2017proportional,aziz2018complexity}. These proportional representation axioms aim to identify cohesive voter groups that collectively approve a small set of candidates and ensure such groups receive proportional representation in the outcome.

Voting rules such as \emph{Proportional Approval Voting} (PAV) \citep{thiele1895om} and \emph{Method of Equal Shares} (MES) \citep{peters2020proportionality,peters2021proportional} satisfy EJR, and EJR implies JR. The core solution is both Pareto optimal and priceable \citep{peters2020proportionality}. Pareto optimality means that no other committee of the same size can make all voters at least as well off while making some strictly better off. Priceability means there exists a payment scheme where each voter pays only for the candidates they approve. The PAV solution is Pareto optimal but not priceable \citep{peters2020proportionality}, whereas the MES solution is priceable but fails to satisfy Pareto optimality \citep{pierczynski2022core}.

Our definition of the discrete Nash core extends the Proportional Approval Voting (PAV) method by incorporating weighted adjustments. These weights serve to reduce the influence of popular candidates on the score function, mitigating the impact on candidates who share voters with them. Similarly, our definition of the weak Nash core modifies the Method of Equal Shares (MES) by introducing a weighted distribution. Instead of allocating a committee member’s payment uniformly among approving voters, our approach distributes it inversely proportional to each voter's weighted utility, ensuring a fairer allocation.

\subsubsection{Discrete-Core Nonemptiness}

Whether every approval-based committee election admits a committee in the discrete core remains open, but exact nonemptiness is known under several restrictions. \citet{pierczynski2022core} prove nonemptiness for several restricted domains, including voter-interval and candidate-interval profiles. \citet{brill2024approval} establish nonemptiness for approval-based apportionment, where multiple seats may be assigned to each approved party.

Other results bound numerical parameters of the election. \citet{peters2025core} proves that the core is nonempty whenever the committee size satisfies $k\leq8$, regardless of the numbers of voters and candidates, and whenever the number of candidates satisfies $m\leq15$. These restrictions are incomparable with a bound on the number of voters.

More recently, \citet{becker2026core} prove core nonemptiness for elections with at most seven voters. Their proof applies to weighted voters and therefore also covers elections with at most seven distinct approval types and arbitrary multiplicities. For elections with at most five voters, they additionally give a polynomial-time algorithm for computing a core committee.

Our Theorem~\ref{thm:bv-main} proves core nonemptiness for elections with at most eight individual equally weighted voters. Relative to \citet{becker2026core}, the genuinely new case consists of eight equally weighted voters with eight distinct approval sets; a repeated approval set yields at most seven voter types. Conversely, their result allows arbitrary weights and multiplicities across seven types, which our theorem does not. Our proof resolves the remaining eight-voter case through an exact computer-assisted classification of eight-row residual regions.

\subsubsection{Approximations of the Core}

Given the challenges in finding a discrete core, recent studies have explored relaxations of this notion. One approach, as shown by \citep{cheng2020group}, introduces a probabilistic variant, ensuring the core property in expectation via a lottery over committees. Another line of work focuses on deterministic approximations, where Jiang et al. \citep{jiang2020approximately} developed an algorithm providing a 32-group-size approximation, while Fain et al. \citep{fain2018fair} and Munagala et al. \citep{munagala2022approximate} proposed different rounding techniques to approximate the core with varying guarantees. Furthermore, Peters and Skowron \citep{peters2020proportionality} explored constraints on the space of deviations, defining restricted variants such as the core subject to cohesiveness, which aligns with Extended Justified Representation (EJR).

\subsection{Roadmap}

Section~\ref{sec:frac} develops the fractional Nash-core theory. Section~\ref{sec:bv-core-eight} then uses that theory to prove discrete-core nonemptiness for elections with at most eight voters. Section~\ref{sec:disc} presents the complementary Nash-core certificate framework for unrestricted discrete elections, and Section~\ref{sec:experi} reports our computational experiments.

\section{The Fractional Setting}    \label{sec:frac}

In this section, we prove the theorems in the fractional setting.

\subsection{Nash Core Solution is Core Solution} \label{sec:frac:incore}

In this subsection, we prove Theorem \ref{thm:frac:incore}.

\fracincore*

\begin{proof}
  Assume $(w,x)$ is a fractional Nash core solution, but $x$ is not a fractional core solution. Then there exists $S\subseteq N$ and $y\in [0,1]^{m}\setminus\{0\}$ such that $\onenorm{y}\leq \frac{k|S|}{n}$, and $\sum_{c\in A_v}x_c< \sum_{c\in A_v}y_c$ for all $v\in S$.

  Let $W = \{c\in C: w_c<1\}$, so $x_c=1$ for $c\in W$. For $v\in S$, since $\sum_{c\in A_v}x_c< \sum_{c\in A_v}y_c$, we have
  \begin{align*}
     \sum_{c\in A_v}w_{c}y_{c} - \sum_{c\in A_v}w_{c}x_{c}
    =& \sum_{c\in A_v\cap W} w_{c} (y_{c} - 1) + \sum_{c\in A_v\setminus W} (y_{c}-x_{c})   \\
    \geq & \sum_{c\in A_v\cap W} (y_{c} - 1) + \sum_{c\in A_v\setminus W} (y_{c}-x_{c})   \\
    =& \sum_{c\in A_v}y_c - \sum_{c\in A_v}x_c > 0,
  \end{align*}
  implying $\sum_{c\in A_v}w_{c}x_{c} < \sum_{c\in A_v}w_{c}y_{c}$.

  Recall that $z=x$ is an optimal solution of the following optimization problem.
\[
  \max_{z \ge 0,\, \|z\|_1 \le k} f(z):=\sum_{v\in N} \log\Bigl(\sum_{c\in A_v} w_c z_c\Bigr).
\]

  Note that $\frac{k}{\onenorm{y}}y$ is also a feasible solution of this optimization problem. Let $\ell = \frac{k}{\onenorm{y}}y - x$. Since the feasible region of the optimization problem is convex, we have $\frac{\partial f}{\partial \ell}(x) \leq 0$ by the optimality of $z=x$. Since we have
  \begin{align*}
     & y^T \cdot \nabla f(x)
    =  \sum_{c\in C}y_{c} \sum_{v\in B_c}\frac{w_c}{\sum_{c'\in A_v}w_{c'}x_{c'}}
    =   \sum_{v\in N} \frac{\sum_{c'\in A_v}w_{c'}y_{c'}}{\sum_{c'\in A_v}w_{c'}x_{c'}}
    \geq   \sum_{v\in S} \frac{\sum_{c'\in A_v}w_{c'}y_{c'}}{\sum_{c'\in A_v}w_{c'}x_{c'}}  > |S|,    \\
    & x^T \cdot \nabla f(x)
    =  \sum_{c\in C}x_c \sum_{v\in B_c}\frac{w_c}{\sum_{c'\in A_v}w_{c'}x_{c'}}
    =  \sum_{v\in N} \frac{\sum_{c'\in A_v}w_{c'}x_{c'}}{\sum_{c'\in A_v}w_{c'}x_{c'}} = n,
  \end{align*}
  so the partial derivative along $\ell$ is
  \begin{align*}
     & \frac{\partial f}{\partial \ell}(x)
    =  \frac{k}{\onenorm{y}} y^T \cdot \nabla f(x) -  x^T \cdot \nabla f(x)
    > 0,
  \end{align*}
  which contradicts to $\frac{\partial f}{\partial \ell}(x) \leq 0$.

  As a result, $x$ is a fractional core solution.
\end{proof}

\subsection{Existence of Nash Core Solutions}   \label{sec:frac:exist}

In this subsection, we prove Theorem \ref{thm:frac:exist}.

\fracexist*

To prove this theorem, we construct a concave game---where a pure Nash equilibrium always exists---and then demonstrate that this equilibrium corresponds to a Nash core solution. We begin by defining the concave game.

\begin{definition}[concave game]
Let $G = \bigl(\{S_i\}_{i=1}^q, \{u_i\}_{i=1}^q\bigr)$ be a $q$-player game in normal form, where:
\begin{enumerate}
    \item Each strategy set $S_i$ is a nonempty, convex, and compact subset of $\mathbb{R}^{d_i}$, where $d_i$ is a positive integer.
    \item Each payoff (utility) function $\theta_i: S_1 \times \cdots \times S_q \to \mathbb{R}$ is a continuous function.
\end{enumerate}
We say $G$ is a \emph{concave game} if, for every player $i$ and every fixed strategy profile $x_{-i} \in \prod_{j \neq i} S_j$, the function
\[
x_i \;\mapsto\; \theta_i(x_i, x_{-i})
\]
is \emph{concave} on $S_i$.
\end{definition}

The concave game has the following property.

\begin{lemma}[\citep{rosen1965existence}] \label{lem:pureexist}
    In the concave game, a pure Nash equilibrium exists.
\end{lemma}

We construct a sequence of concave games with $m+1$ players. Fix a positive integer $t$ and consider the game $G^{(t)}$. We use $w=(w_c)_{c\in C}$ and $x=(x_c)_{c\in C}$ for a generic strategy profile in this game. Each candidate $c\in C$ corresponds to a player who chooses an action $w_c\in \left[\frac{1}{t+1},1\right]$, while the additional Player~$0$ chooses an action from $\Omega_x := \{x\in [0,k]^m: \sum_{c\in C}x_c\leq k\}$. On strategy variables, the superscript $(t)$ will be reserved for an equilibrium profile $(w^{(t)},x^{(t)})$ selected from $G^{(t)}$. Player~$0$'s utility is defined as the geometric mean of the weighted voter utilities:
\begin{equation*}
\theta^{(t)}_0(w,x)
:=
\left(
\prod_{v\in N}\sum_{c\in A_v}w_c x_c
\right)^{1/n}.
\end{equation*}
The utility of candidate player $c\in C$ at a generic strategy profile $(w,x)$ is $\theta^{(t)}_c(w,x)=-(1-x_c)(1-w_c)$.

The geometric mean is continuous, coordinatewise nondecreasing, and concave on $\mathbb{R}^n_+$. Since the vector of weighted voter utilities is continuous in the full strategy profile $(w,x)$ and affine in $x$ for fixed $w$, Player~$0$'s payoff is continuous in the full strategy profile and concave in her own strategy. Each candidate player's payoff is continuous in the full strategy profile and affine in the candidate's own strategy $w_c$. Therefore, $G^{(t)}$ is a concave game.

Moreover, Assumption~\ref{assup:1} implies that $A_v\neq\emptyset$ for every voter $v$. Since every feasible candidate strategy satisfies $w_c\geq 1/(t+1)>0$, the choice $x_c=k/m$ for every $c\in C$ belongs to $\Omega_x$ and gives every voter positive weighted utility. Hence the maximum payoff of Player~$0$ is positive, and every best response of Player~$0$ gives every voter positive weighted utility. At every such point,
\begin{equation*}
\log \theta^{(t)}_0(w,x)
=
\frac{1}{n}\sum_{v\in N}
\log\left(\sum_{c\in A_v}w_c x_c\right).
\end{equation*}
Therefore, Player~$0$ has exactly the same best responses as in the candidate-weighted Nash-social-welfare maximization problem.

By Lemma \ref{lem:pureexist}, $G^{(t)}$ has a pure Nash equilibrium; select one and denote it by $(w^{(t)},x^{(t)})$. Thus, $w$ and $x$ denote generic strategies, whereas $w^{(t)}$ and $x^{(t)}$ denote the selected equilibrium strategies in $G^{(t)}$. Since for any $t\geq 1$, $w^{(t)}\in [0,1]^m$ and $x^{(t)}\in \Omega_x$, both are in a compact set. By the Bolzano–Weierstrass theorem \citep{bartle2000introduction}, there is a subsequence $\{(w^{(r_i)},x^{(r_i)})\}_{i=1}^{\infty}$ that converges to a point $(w^{(\infty)}, x^{(\infty)})\in [0,1]^m\times \Omega_x$, where $\{r_i\}_{i=1}^{\infty}$ is a strictly increasing sequence of positive integers.

We show that $w^{(\infty)}$ is actually strictly positive.

\begin{lemma}
    We have $w^{(\infty)}\in (0,1]^m$.
\end{lemma}

\begin{proof}

Define $W=\{c: w^{(\infty)}_c = 0\}$. We assume $W$ is not empty, and define $S=\{v\in N: A_v\subseteq W\}$. By assumption 1, we have $\frac{n}{k}|W|>|S|$.

W.l.o.g., for $c\in W$, assume $w_c^{(r_i)}<1$ for any $i\geq 1$, otherwise we can remove the first a few entries of $\{r_i\}_{i=1}^{\infty}$ to achieve this property because $\{w_c^{(r_i)}\}_{i=1}^{\infty}$ converges to 0. Note that $w_c^{(r_i)}<1$ is the optimal strategy of Player $c$ in response to other players' strategies, so $x_c^{(r_i)}\geq 1$, otherwise $w_c^{(r_i)}=1$ is a better strategy for Player $c$.

In game $G^{(r_i)}$, since $x^{(r_i)}$ is an optimal strategy of Player 0 in response to other players' strategies, $z=x^{(r_i)}$ is an optimal solution of
\[
\max_{z \ge 0,\, \|z\|_1\le k}\;\; \sum_{v\in N} \log\Bigl(\sum_{c\in A_v} w_c^{(r_i)} z_c\Bigr).
\]

Recall that in the proof of Lemma \ref{lem:frac:equiv:forward}, we showed Equation \ref{eq:deriv} and Inequality \ref{ieq:deriv}. Using similar KKT-condition arguments, we have
\begin{align} \label{ieq:exist:deriv}
\sum_{v : c\in A_v} \frac{w^{(r_i)}_c}{\sum_{c'\in A_v}w^{(r_i)}_{c'} x^{(r_i)}_{c'}}
  \;\leq \;\frac{n}{k},
\end{align}
and equality holds when $x^{(r_i)}_c>0$. Since $x^{(r_i)}_c\geq 1$ for $c\in W$, we have
\begin{align}
\frac{n}{k}|W|
\leq & \sum_{c\in W}\frac{n}{k} x^{(r_i)}_c  \label{ieq:exist:1} \\
= & \sum_{c\in W} \sum_{v : c\in A_v} \frac{w^{(r_i)}_c x^{(r_i)}_c}{\sum_{c'\in A_v}w^{(r_i)}_{c'} x^{(r_i)}_{c'}}    \\
= &  \sum_{v\in N} \frac{\sum_{c\in A_v\cap W} w^{(r_i)}_c x^{(r_i)}_c}{\sum_{c'\in A_v}w^{(r_i)}_{c'} x^{(r_i)}_{c'}}    \\
\leq & \left( |S| + \sum_{v\in N\setminus S} \frac{\sum_{c\in A_v\cap W} w^{(r_i)}_c x^{(r_i)}_c}{\sum_{c'\in A_v}w^{(r_i)}_{c'} x^{(r_i)}_{c'}}\right).
\end{align}

For $v\in N\setminus S$, we have
\begin{align}   \label{ieq:exist:2}
\lim\limits_{i\rightarrow \infty}\left( \sum_{c'\in A_v}w^{(r_i)}_{c'} x^{(r_i)}_{c'} \right)>0,
\end{align}
otherwise for $c\in A_v\setminus W$, since $\lim\limits_{i\rightarrow \infty} w^{(r_i)}_c>0$,
\[
\lim\limits_{i\rightarrow \infty} \sum_{v' : c\in A_{v'}} \frac{w^{(r_i)}_c}{\sum_{c'\in A_{v'}}w^{(r_i)}_{c'} x^{(r_i)}_{c'}}
\geq \lim\limits_{i\rightarrow \infty} \frac{w^{(r_i)}_c}{\sum_{c'\in A_v}w^{(r_i)}_{c'} x^{(r_i)}_{c'}} = +\infty,
\]
which violates Inequality \ref{ieq:exist:deriv}.

Combining formulas \ref{ieq:exist:1} to \ref{ieq:exist:2}, we have
\begin{align*}
\frac{n}{k}|W|
\leq & |S| +  \lim\limits_{i\rightarrow \infty} \sum_{v\in N\setminus S} \frac{\sum_{c\in A_v\cap W} w^{(r_i)}_c x^{(r_i)}_c}{\sum_{c'\in A_v}w^{(r_i)}_{c'} x^{(r_i)}_{c'}}    \\
= & |S| +   \sum_{v\in N\setminus S} \frac{\lim\limits_{i\rightarrow \infty} \sum_{c\in A_v\cap W} w^{(r_i)}_c x^{(r_i)}_c}{\lim\limits_{i\rightarrow \infty} \sum_{c'\in A_v}w^{(r_i)}_{c'} x^{(r_i)}_{c'}}
=|S|,
\end{align*}
which contradicts to Assumption 1. As a result, we have $w^{(\infty)}\in (0,1]^m$.

\end{proof}

Since $w^{(\infty)}\in (0,1]^m$, there is a large enough $i$ such that $w^{(r_i)}_c>\frac{1}{r_i+1}$ for each $c\in C$. Since $w^{(r_i)}_c$ is the best strategy of Player $c$ in response to other players' strategies, we have $x^{(r_i)}_c\leq 1$, otherwise $w^{(r_i)}_c=\frac{1}{r_i+1}$ is a better strategy for Player $c$. Furthermore, based Definition \ref{def:nashcore}, $(w^{(r_i)}, x^{(r_i)})$ satisfies weight-selection coupling ($x_c=1$ if $w_c<1$) by the optimality of $w^{(r_i)}_c$ for Player $c$, and it satisfies Nash optimality by the optimality of $x^{(r_i)}_c$ for Player $0$. As a result, $(w^{(r_i)}, x^{(r_i)})$ is in the fractional Nash core, completing the proof of Theorem \ref{thm:frac:exist}.

\subsection{Connection with the Convex Program}\label{sec:compute}

Kroer and Peters~\citep{kroer2025computing} introduce a convex optimization problem whose optimal solution lies in the fractional core. For the approval-based setting under Assumption \ref{assup:1}, an equivalent formulation that is slightly simpler is as follow:
\begin{equation}\label{eq:primal}\tag{$\mathcal P$}
  \begin{alignedat}{3}
    \max_{b,x}\quad &-\sum_{v\in N}\sum_{c\in A_v} b_{vc}\,\log\!\Bigl(\tfrac{b_{vc}}{x_c}\Bigr)\\[2pt]
    \text{s.t.}\quad &\sum_{c\in A_v} b_{vc}=1           &&\qquad(\forall v\in N),\\
                      &\sum_{v\,:\,c\in A_v} b_{vc}=\tfrac{n}{k}\,x_c &&\qquad(\forall c\in C),\\
                      &b_{vc}\ge 0                           &&\qquad(\forall v,c),\\
                      &0\le x_c\le 1                         &&\qquad(\forall c\in C).
  \end{alignedat}
\end{equation}

The variables $x_c$ denote the total funding allocated to each candidate $c$, while $b_{vc}$ specifies how much voter $v$ contributes to candidate $c$. Each voter has a unit budget, and the price per unit for each candidate is fixed at $n/k$.

The objective is \emph{strictly concave} in $b$ and affine in $x$, while all constraints are affine; hence~\eqref{eq:primal} is a convex program, and the KKT conditions are sufficient for optimality.

\begin{lemma}[\citep{kroer2025computing}]
  Every optimal solution of~\eqref{eq:primal} lies in the fractional core.
\end{lemma}

We refine this statement and link it to proportional payments.

\begin{theorem}[Restatement of Theorem~\ref{thm:compute}]
  Let $(w,x)$ be a \emph{fractional committee} supported by proportional payments and define
  \[
    \phi_v(x)\;:=\;\sum_{c\in A_v} w_c\,x_c,
    \qquad
    b_{vc}\;:=\;
      \frac{w_c\,x_c}{\phi_v(x)}\,
      \mathbf 1_{\{c\in A_v\}}\!.
  \]
  Then $(b,x)$ is an optimal solution to the primal program~\eqref{eq:primal}.
\end{theorem}

\begin{proof}
Since~\eqref{eq:primal} is a convex program that admits a strictly feasible point, the Karush–Kuhn–Tucker (KKT) conditions are necessary and sufficient for optimality.

\smallskip
\noindent\textit{Dual multipliers.} For each voter $v\in N$ and candidate $c$ let
\begin{itemize}
  \item $\lambda_v\in \real$ correspond to the constraint $\sum_{c\in A_v} b_{vc}=1$,
  \item $\mu_c\in \real$ correspond to the constraint $\sum_{v:\,c\in A_v} b_{vc}= \tfrac{n}{k}\,x_c$,
  \item $\pi_{vc}\ge 0$ correspond to $b_{vc}\ge 0$,
  \item $\nu_c\ge 0$ correspond to the upper bound $x_c\le 1$.
\end{itemize}

\smallskip
\noindent\textit{Choice of variables.} Using the proportional–payment solution $(w,x)$, set
\[
  b_{vc}:=\frac{w_c\,x_c}{\phi_v(x)}\,\mathbf 1_{\{c\in A_v\}},\qquad
  \lambda_v:=\log\!\bigl(\phi_v(x)\bigr),\qquad
  \mu_c:=1+\log w_c,
\]
\[
  \nu_c:=-\frac{n}{k}\,\log w_c\quad(\ge 0\text{ because }w_c\le 1),
  \qquad
  \pi_{vc}:=0 .
\]

\smallskip
\noindent\textit{Stationarity.} Let $\mathcal L$ be the Lagrangian. Differentiating w.r.t.\ $b_{vc}$ and $x_c$ (for $c\in A_v$) gives
\begin{align*}
  \frac{\partial\mathcal L}{\partial b_{vc}}: &
  \quad
  \log\!\Bigl(\tfrac{b_{vc}}{x_c}\Bigr)+1+\lambda_v-\mu_c-\pi_{vc}=0, \\[4pt]
  \frac{\partial\mathcal L}{\partial x_c}: &
  \quad
  \mu_c+\frac{k}{n}\,\nu_c=1 .
\end{align*}
Substituting the chosen multipliers verifies both equalities, so stationarity holds.

\smallskip
\noindent\textit{Complementary slackness.} Because $\pi_{vc}=0$, the condition $\pi_{vc}b_{vc}=0$ is immediate. For the upper‑bound constraint,
\[
  \nu_c(1-x_c)=0\quad\Longleftrightarrow\quad
  \bigl[-\tfrac{n}{k}\log w_c\bigr]\,(1-x_c)=0 .
\]
Hence either $x_c=1$ or $w_c=1$; both cases are consistent with proportional payments and satisfy the required product equality.

\smallskip
\noindent\textit{Feasibility.} Primal feasibility is built into the definitions of $b$ and $x$. Dual feasibility holds because $\nu_c,\pi_{vc}\ge 0$ and all unrestricted multipliers are free.

\smallskip
All KKT conditions are therefore met, so $(b,x)$ is a globally optimal solution of~\eqref{eq:primal}.
\end{proof}

\subsection{Connection with Lindahl equilibrium}   \label{sec:lindahl}

A specialized definition of Lindahl equilibrium for the Participatory Budgeting Problem is provided by \citep{fain2016core}, which always exists \citep{foley1970lindahl}. Building on that, the following further refines the concept for the fractional committee selection problem, which is referred to in \citep{fain2016core} as the ``saturating utility''.

\begin{definition}[saturating version of \citep{fain2016core}]   \label{def:lindahl}
    In a public goods market with budget $k$, per-voter prices $q_1,q_2,\dots,q_n$ each in ${\real^{+}}^{m}$ and allocation $x\in [0,1]^m$ constitute a Lindahl equilibrium if the following two conditions hold:
    \begin{enumerate}
      \item For every agent $v\in N$, the utility $u_v(y_v):=\sum_{c\in A_v}y_{vc}$ is maximized subject to $q_v^{T} y_v\leq \frac{k}{n}$ and $y_v\in [0,1]^m$ when $y_v=x$.
      \item The profit defined as $ (\sum_{v}q_v)^T z - \|z\|_1 $, subject to $z\in [0,1]^m$ and $\|z\|_{1}\leq k$ is maximized when $z=x$.
    \end{enumerate}
\end{definition}

We show that proportional payment implies Lindahl equilibrium.

\begin{theorem} \label{thm:lindahl}
    For $w\in (0, 1]^m$ and $x\in [0,1]^m$ where $\|x\|_1\leq k$, if $(w,x)$ is supported by a proportional payment, then there are per-voter prices $q_1,q_2,\dots,q_n$ each in ${\real^{+}}^{m}$ such that $(q,x)$ constitutes a Lindahl equilibrium.
\end{theorem}

We defer the proof of this theorem to Appendix~\ref{sec:proof:lindahl}.

We note that the converse of this theorem does not hold, as illustrated by the following example, which is the same as Example \ref{exp:multicore}.

\begin{example}
  Consider the fraction case with $n=2$ voters, $m=3$ candidates, and $k=2$ committee members. Voter $v_1$ approves $\{c_1,c_2\}$, and $v_2$ approves $\{c_1,c_3\}$.

  Under a Lindahl Equilibrium, the outcome is $x=(1,1,0)$, supported by price vectors $p_1=(0,1,0), p_2=(1,0,1)$.

  However, the unique proportional payment solution is given by $w=(\frac{1}{2},1,1)$, $x=(1,\frac{1}{2},\frac{1}{2})$.
\end{example}

We next use the fractional Nash-core allocation and its proportional-payment prices to establish an exact discrete-core existence result for elections with at most eight voters.

\providecommand{\R}{\mathbb{R}}
\providecommand{\Z}{\mathbb{Z}}
\providecommand{\one}{\mathbf{1}}
\providecommand{\supp}{\operatorname{supp}}
\providecommand{\cl}{\operatorname{cl}}
\providecommand{\ri}{\operatorname{relint}}
\providecommand{\floor}[1]{\left\lfloor #1\right\rfloor}
\providecommand{\ceil}[1]{\left\lceil #1\right\rceil}
\providecommand{\cD}{\mathcal D}
\providecommand{\cU}{\mathcal U}
\providecommand{\cR}{\mathcal R}
\newif\ifbveightformal
\newcommand{\bveightclassificationtitle}{Informal}

\section{Discrete-Core Nonemptiness for Eight Equal-Weight Voters}
\label{sec:bv-core-eight}

We establish an exact discrete-core existence theorem for elections with at most eight equally weighted individual voters. \citet{becker2026core} prove the more general result that the discrete core is nonempty whenever there are at most seven voter types, allowing arbitrary voter weights. Consequently, an eight-voter election with a repeated approval set is already covered by their theorem. The genuinely additional case considered here therefore consists of eight equally weighted voters with eight distinct approval sets.

Our proof of the eight-voter case starts from a fractional Nash-core allocation and its proportional-payment prices. After choosing a fractional Nash-core pair with the minimum possible number of fractional candidates, the Nash first-order conditions reduce every unresolved instance to a positive-dual full-column-rank binary antichain with at most eight columns. We then verify the remaining eight-row floor regions using exact computer-assisted arguments. The principal new ingredients are fixed and adaptive price--saturation certificates, which strengthen the one-deficit price arguments sufficient in smaller dimensions.

\begin{restatable}{theorem}{bveightmain}
\label{thm:bv-main}
Every approval-based committee election with at most eight equally weighted voters has a committee in the discrete core.
\end{restatable}

The theorem allows arbitrary committee size $k$ and an arbitrary number $m$ of candidates. It is parameterized by the number of individual equal-weight voters, not by the number of voter types with arbitrary multiplicities.

\subsection{Reduction to a residual rounding problem}
\label{sec:bv-reduction}

For a candidate $c$, write $N(c):=\{i\in N:c\in A_i\}$.

\begin{lemma}[Floor rounding]\label{lem:bv-floor-rounding}
Let $x\in[0,1]^m$ lie in the fractional core.  If an integral committee $W\subseteq C$ satisfies $|W|\le k$ and $|W\cap A_i|\ge\floor{u_i(x)}$ for every $i \in N$, then $W$ lies in the discrete core.
\end{lemma}

\begin{proof}
Suppose that an integral committee $T$ blocks $W$ through a coalition $S$. For every $i\in S$, integrality gives
\[
  |T\cap A_i|\ge |W\cap A_i|+1
  \ge \floor{u_i(x)}+1>u_i(x).
\]
The same committee $T$, viewed as a fractional committee, therefore blocks $x$ through $S$, contradicting fractional-core stability.
\end{proof}

\begin{lemma}[Reduction when Assumption~\ref{assup:1} fails]
\label{lem:bv-assumption-one-reduction}
Let $S\subsetneq N$ be nonempty, put $U:=\bigcup_{i\in S}A_i$, and suppose $|U|\le k|S|/n$.  Define $N':=N\setminus S$, $C':=C\setminus U$, and $k':=k-|U|$. If the residual election $(N',C',k')$ has a core committee $W'$, then $U\cup W'$ is a core committee of the original election.
\end{lemma}

\begin{proof}
Every voter in $S$ receives every candidate she approves and therefore cannot strictly improve.  Moreover,
\[
  \frac{k'}{|N'|}=\frac{k-|U|}{n-|S|}\ge\frac{k}{n}.
\]
Suppose that a coalition $Q\subseteq N'$ blocks $U\cup W'$ through a committee $T$.  Put $T':=T\setminus U$.  For every $i\in Q$,
\begin{align*}
 |T'\cap A_i|-|W'\cap A_i|
 &=|T\cap A_i|-|(U\cup W')\cap A_i|\\
 &\quad+|U\cap A_i|-|(T\cap U)\cap A_i|>0.
\end{align*}
Also,
\[
  |T'|\le |T|\le\frac{k|Q|}{n}
  \le\frac{k'|Q|}{|N'|}.
\]
Thus $T'$ blocks $W'$ in the residual election, a contradiction.
\end{proof}

Let $C^+:=\bigcup_{i\in N}A_i$.  Candidates outside $C^+$ may be deleted.  If $|C^+|\le k$, selecting every candidate in $C^+$ is core-stable. Thus, under Assumption~\ref{assup:1}, only the case $|C^+|>k$ remains.

By Theorem~\ref{thm:frac:exist}, a fractional Nash-core pair exists. Among all such pairs, choose one whose allocation has the minimum possible number of fractional coordinates. Theorem~\ref{thm:frac:incore} implies that $x$ lies in the fractional core, while Theorem~\ref{thm:frac:equiv} and the positive-utility condition in Definition~\ref{def:payment} imply that every $\phi_i(x)$ is positive. The Nash optimality condition in Definition~\ref{def:nashcore} implies $\|x\|_1=k$: otherwise scaling $x$ to use the full budget would strictly increase every weighted utility.

Define $t:=k/n$, $\alpha_i:=t/\phi_i(x)$, and $\alpha(R):=\sum_{i\in R}\alpha_i$. The affordability and stability conditions in Definition~\ref{def:payment} give, for every candidate $c$,
\begin{equation}\label{eq:bv-nash-price}
  w_c\alpha(N(c))=1\quad\text{if }x_c>0,
  \qquad
  w_c\alpha(N(c))\le1\quad\text{if }x_c=0.
\end{equation}
The weight--selection coupling in Definition~\ref{def:payment} also gives
\begin{equation}\label{eq:bv-price-by-status}
  \alpha(N(c))=1\ \text{if }0<x_c<1,
  \qquad
  \alpha(N(c))\le1\ \text{if }x_c=0.
\end{equation}

Put $I:=\{c:x_c=1\}$, $F:=\{c:0<x_c<1\}$, and $O:=\{c:x_c=0\}$. Let $d:=|F|$ and identify the candidates in $F$ with $[d]$.  Let $A\in\{0,1\}^{n\times d}$ be the voter-by-fractional-candidate incidence matrix, so $A_{ic}=1$ exactly when voter $i$ approves fractional candidate $c$. Let $f:=x_F\in(0,1)^d$, $\kappa:=\one^\top f=k-|I|\in\Z_{\ge0}$, and $r:=Af$. Write $h:=\floor{r}\in\Z^n_{\ge0}$ and $\delta:=r-h$, so $0\le\delta_i<1$. For each voter $i$, let $g_i:=|I\cap A_i|$ and $\beta_i:=\sum_{c\in I\cap A_i}w_c$. Every fractional candidate has weight one, and therefore
\begin{equation}\label{eq:bv-beta-relation}
  \alpha_i(\beta_i+r_i)=t,
  \qquad
  \beta_i=\frac{t}{\alpha_i}-r_i,
  \qquad
  g_i\ge\beta_i.
\end{equation}

For $z\in\{0,1\}^d$, call $Az$ its residual utility vector.  We say that the integer vector $h$ is \emph{implementable at budget $\kappa$} if there exists $z\in\{0,1\}^d$ with $\one^\top z\le\kappa$ and $Az\ge h$.  Such a selection may be padded to exactly $\kappa$ residual candidates without decreasing any utility.

If $h$ is implementable, set $W:=I\cup\supp(z)$.  Then $|W|\le k$ and
\[
  |W\cap A_i|=g_i+(Az)_i
  \ge g_i+h_i
  =\floor{g_i+r_i}
  =\floor{u_i(x)}.
\]
Lemma~\ref{lem:bv-floor-rounding} then gives a discrete-core committee.  We may therefore restrict attention to nonimplementable $h$.

\begin{restatable}[Minimum-fractionality reduction]{proposition}{bvreducedmatrix}
\label{prop:bv-reduced-matrix}
If $h$ is not implementable, then:
\begin{enumerate}[label=\textup{(\roman*)}]
\item the columns of $A$ are pairwise distinct and form an antichain under support inclusion;
\item $A$ has full column rank, and hence $d\le n$;
\item $2\le\kappa\le d-2$, and hence $4\le d\le n$;
\item $A^\top\alpha=\one$ for some $\alpha\in\R^n_{>0}$.
\end{enumerate}
\end{restatable}

\begin{proof}[Proof sketch]
For part~\textup{(i)}, identical supports can be merged by transferring mass between the two coordinates until one becomes integral, without changing the budget or any voter utility.  If one support is properly contained in another, transferring a small amount of mass toward the larger support weakly improves all weighted utilities and strictly improves some of them, contradicting Nash optimality.

For part~\textup{(ii)}, a nonzero vector in the kernel of $A$ gives a direction along which all voter utilities are unchanged.  Orienting this direction to have nonpositive total mass and moving until a coordinate reaches the boundary either reduces the number of fractional coordinates or produces a Nash optimum that does not exhaust the budget; the latter could be scaled up, again a contradiction.  Hence $A$ has full column rank.  The endpoint cases $\kappa\in\{0,1,d-1,d\}$ are directly implementable, which yields part~\textup{(iii)}.  Finally, every fractional candidate has weight one, so the Nash-price identity gives $\alpha(N(c))=1$ for every column, equivalently $A^\top\alpha=\one$ with $\alpha>0$.  Full details are given in Appendix~\ref{sec:proof:bv-reduced-matrix}.
\end{proof}

For a fixed residual matrix $A$, define its positive-dual domain by
\begin{equation}\label{eq:bv-dual-domain}
  \cD(A):=\{\alpha\in\R^n_{>0}:A^\top\alpha=\one\}.
\end{equation}
The exact classification therefore needs to enumerate only full-column-rank binary antichains with $4\le d\le n$, $2\le\kappa\le d-2$, and $\cD(A)\ne\varnothing$.

For each such matrix and residual budget, we compare a fractional allocation with the utility vectors attainable by integral selections. Fix a full-column-rank antichain $A\in\{0,1\}^{n\times d}$ and an integer residual budget $\kappa$, and define $\cU(A,\kappa):=\{Az:z\in\{0,1\}^d,\ \one^\top z=\kappa\}$ to be the finite set of utility vectors obtainable by selecting exactly $\kappa$ residual candidates. An integer vector $h\ge0$ is implementable exactly when some $u\in\cU(A,\kappa)$ satisfies $u\ge h$. A nonimplementable vector $h$ is coordinatewise minimal if no smaller nonnegative integer vector is nonimplementable.

For an integer utility profile $h$, we define the fractional allocation region indexed by $h$ and its closed relaxation to be
\begin{align*}
 \cR_{<}(A,\kappa,h)
 &:={}
 \{f\in(0,1)^d:
 \one^\top f=\kappa,
 \ h\le Af<h+\one\},\\
 \cR_{\leq}(A,\kappa,h)
 &:={}
 \{f\in[0,1]^d:
 \one^\top f=\kappa,
 \ h\le Af\le h+\one\}.
\end{align*}
All vector inequalities are componentwise.

\subsection{Price and saturation certificates}
\label{sec:bv-rounding-certificates}

We first derive two conditions that any coalition blocking a residual committee must satisfy.

Fix a residual integral committee $z\in\{0,1\}^d$ with $\one^\top z=\kappa$, and put $q:=Az$ and $W:=I\cup\supp(z)$. Thus, $q_i=(Az)_i$ is the number of selected candidates from $F$ that voter $i$ approves, and $W$ is an integral committee.

\begin{restatable}[Necessary conditions for a blocking coalition]{proposition}{bvblockerconditions}\label{prop:bv-blocker-conditions}
If an integral committee $T$ blocks $W$ through a coalition $S$ of size $s$, then
  \begin{enumerate}[label=(\roman*)]
  \item $\sum_{i\in S}\alpha_i(q_i+1-r_i)\le0$;
  \item Every $i\in S$ satisfies $\ceil{t/\alpha_i+q_i-r_i}+1\le \floor{ts}$.
  \end{enumerate}
\end{restatable}

\begin{proof}[Interpretation and proof sketch]
Part~\textup{(i)} says that a blocking coalition must have nonpositive net Nash-price margin after accounting for candidates added to and removed from the fractional benchmark.  To see this, write $R:=T\setminus I$ for the new candidates and $D:=I\setminus T$ for the fully selected candidates that are removed.  Nash prices bound the coalition's total price for $R$ by $|R|$, while its normalized payments toward $I$ give a lower bound on the payments lost through $D$.  Combining these bounds with $|T|\le t|S|$ and comparing with the utility gain $u_i(R)-u_i(D)\ge q_i+1$ yields part~\textup{(i)}.

Part~\textup{(ii)} is the corresponding individual saturation condition.  By \eqref{eq:bv-beta-relation}, voter $i$ already receives at least $\ceil{t/\alpha_i-r_i}+q_i$ approved candidates from the incumbent committee. Strict improvement therefore requires one additional approved candidate, and this demand cannot exceed the coalition's cardinality entitlement $\floor{ts}$.  The complete payment accounting is given in Appendix~\ref{sec:proof:bv-blocker-conditions}.
\end{proof}

For a coalition $S$ and a residual committee $z$, define the price margin
\[
  P_{S,z}(\alpha,f)
  :=\sum_{j\in S}\alpha_j\bigl((Az)_j+1-(Af)_j\bigr).
\]
The price filter says that a blocking coalition must have nonpositive price margin.

These conditions give two ways to construct a core committee from the fractional allocation.

\begin{definition}\label{def:bv-rounding-vectors}
Let $h$ be nonimplementable, and define
\[
  B(h):=\{i:h_i>0\text{ and }h-e_i\text{ is implementable}\}.
\]
For $i\in B(h)$, let $z^i\in\{0,1\}^d$ denote any vector satisfying $\one^\top z^i=\kappa$ and $Az^i\ge h-e_i$. Necessarily $(Az^i)_i=h_i-1$ and $(Az^i)_j\ge h_j$ for every $j\ne i$, since otherwise $Az^i\ge h$ would implement $h$.
\end{definition}

\begin{theorem}\label{thm:bv-rounding-criterion}
Let $i\in B(h)$, and let $z^i$ satisfy the conditions in Definition~\ref{def:bv-rounding-vectors}. If
\begin{align}
 \alpha_i\delta_i&<t(1-\alpha_i),\label{eq:bv-alpha-singleton}\\
 \alpha_i\delta_i&<\alpha_j(1-\delta_j)
 \quad\forall j\in N\setminus\{i\},\label{eq:bv-alpha-pairwise}
\end{align}
then $I\cup\supp(z^i)$ lies in the discrete core.
\end{theorem}

\begin{proof}
Let $q:=Az^i$ and suppose a coalition $S$ blocks the associated committee.  If $i\notin S$, then $q_j+1-r_j\ge1-\delta_j>0$ for every $j\in S$, so $P_{S,z^i}(\alpha,f)>0$, contradicting Proposition~\ref{prop:bv-blocker-conditions}(i).  Hence $i\in S$.

At voter $i$ the price coefficient is $q_i+1-r_i=-\delta_i$, while every other coefficient is at least $1-\delta_j$.  The price filter therefore implies $\sum_{j\in S\setminus\{i\}}\alpha_j(1-\delta_j)\le\alpha_i\delta_i$. If $|S|\ge2$, this contradicts \eqref{eq:bv-alpha-pairwise}.

It remains to rule out $S=\{i\}$. Condition \eqref{eq:bv-alpha-singleton} is equivalent to $t/\alpha_i-\delta_i>t$. Since $q_i-r_i=-1-\delta_i$, we have $t/\alpha_i+q_i-r_i>t-1$. Consequently, $\ceil{t/\alpha_i+q_i-r_i}+1>\floor t$, contradicting Proposition~\ref{prop:bv-blocker-conditions}(ii) with $s=1$.
\end{proof}

\begin{corollary}\label{cor:bv-integer-coordinate}
Suppose $i\in B(h)$ and $(Af)_i=h_i$.  Then a core committee exists.
\end{corollary}

\begin{proof}
Choose $z^i$ as in Definition~\ref{def:bv-rounding-vectors}.  Here $\delta_i=0$.  Since $h_i>0$, voter $i$ is incident with a fractional column. If the only such support were the singleton $\{i\}$, distinctness and antichainness would force every other fractional support to omit $i$, making $r_i$ a single fractional coordinate rather than the positive integer $h_i$. Thus some fractional column contains $i$ and another voter.  Its total $\alpha$-mass is one, so positivity of $\alpha$ gives $\alpha_i<1$. Consequently \eqref{eq:bv-alpha-singleton} holds because $\delta_i=0$, and \eqref{eq:bv-alpha-pairwise} holds because $\alpha_j>0$ and $\delta_j<1$ for every $j\ne i$.  Theorem~\ref{thm:bv-rounding-criterion} applies.
\end{proof}

For eight voters, the two inequalities used in Theorem~\ref{thm:bv-rounding-criterion} do not cover every residual floor region.  We therefore derive explicit lower bounds on $\alpha_i$ for members of a blocking coalition and then test coalition price margins on the resulting restricted domain.

We use Proposition~\ref{prop:bv-blocker-conditions}(ii) to derive a necessary condition that depends only on the residual fractional part containing $A,\kappa$, and $h$, but not the original budget $k$ or the fully selected candidates $I$.

For integers $\kappa\ge1$, $1\le s\le8$, and $a\in\Z$, define
\begin{equation}\label{eq:bv-lambda-def}
 \lambda_{\kappa,s,a}:=
 \min\left(
 \left\{\frac1s\right\}
 \cup
 \left\{
 \frac{B}{8(\floor{Bs/8}+a)}:
 B=\kappa,\ldots,\kappa+7,
 \ \floor{Bs/8}+a>0
 \right\}
 \right).
\end{equation}
The set in \eqref{eq:bv-lambda-def} is nonempty because it contains $1/s$.

\begin{restatable}{lemma}{bvlambda}
  \label{lem:bv-lambda}
Consider an eight-voter election.  If a coalition $S$ of size $s$ blocks a residual committee $z$, then every $i\in S$ satisfies $\alpha_i>\lambda_{\kappa,s,h_i-(Az)_i}$.
\end{restatable}

Indeed, Proposition~\ref{prop:bv-blocker-conditions}(ii), with $a:=h_i-(Az)_i$, first gives $\alpha_i>k/[8(\floor{ks/8}+a)]$.  Minimizing this expression over all possible original budgets $k\ge\kappa$ can be done separately in the eight residue classes modulo eight.  Within each class the ratio is monotone and converges to $1/s$, so its infimum is either its value at the first representative in $\{\kappa,\ldots,\kappa+7\}$ or $1/s$.  This is precisely the threshold in \eqref{eq:bv-lambda-def}; the complete derivation is given in Appendix~\ref{sec:proof:bv-lambda}.

For $i\in B(h)$ and a vector $z^i$ satisfying $\one^\top z^i=\kappa$ and $Az^i\ge h-e_i$, define
\begin{equation*}
 G_i(\alpha,f):=
 \frac{\kappa}{8}(1-\alpha_i)
 -\alpha_i\bigl((Af)_i-h_i\bigr)
 =\frac{\kappa}{8}(1-\alpha_i)-\alpha_i\delta_i.
\end{equation*}

\begin{theorem}
  \label{thm:bv-eight-criterion}
Consider an eight-voter election.
\begin{enumerate}[label=(\roman*)]
\item Let $i\in B(h)$ and let $z^i\in\{0,1\}^d$ satisfy $\one^\top z^i=\kappa$ and $Az^i\ge h-e_i$. Suppose $G_i(\alpha,f)>0$ and, for every coalition $S\ni i$ with $|S|\ge2$, the inequalities $\alpha_j>\lambda_{\kappa,|S|,h_j-(Az^i)_j}$ for every $j\in S$ imply $P_{S,z^i}(\alpha,f)>0$. Then $I\cup\supp(z^i)$ lies in the discrete core.

\item Let $E\subseteq[8]$ be nonempty. Suppose that for every $i\in E$ there is a vector $z^i\in\{0,1\}^d$ satisfying $\one^\top z^i=\kappa$, $Az^i\ge h-e_i$, and the coalition-margin implication in part~(i), and suppose
\begin{equation}\label{eq:bv-adaptive-sum}
 \sum_{i\in E}G_i(\alpha,f)>0.
\end{equation}
Then at least one of the committees $I\cup\supp(z^i)$, $i\in E$, lies in the discrete core.
\end{enumerate}
\end{theorem}

\begin{proof}
For part~(i), note first that $(Az^i)_j+1-r_j\ge1-\delta_j>0$ for every $j\ne i$, whereas the coefficient at $i$ is $-\delta_i\le0$.  Hence any coalition satisfying the necessary price inequality must contain $i$.

A singleton block is impossible. Indeed, $G_i>0$ implies $\alpha_i<1$ and $\alpha_i\delta_i<(\kappa/8)(1-\alpha_i)\le t(1-\alpha_i)$, so the singleton argument in Theorem~\ref{thm:bv-rounding-criterion} applies. If a blocking coalition has size at least two, Lemma~\ref{lem:bv-lambda} places its price vector in the stipulated strict domain, and the resulting positive price margin contradicts Proposition~\ref{prop:bv-blocker-conditions}(i).

For part~(ii), equation \eqref{eq:bv-adaptive-sum} gives $G_i(\alpha,f)>0$ for some $i\in E$. Part~(i) applies to the corresponding vector $z^i$. The successful index may depend on the actual pair $(\alpha,f)$; the conclusion is existential.
\end{proof}

We call the data in Theorem~\ref{thm:bv-eight-criterion}(i) a \emph{fixed price--saturation certificate}.  The family of one-deficit roundings and the positive aggregate margin in part~\textup{(ii)} form an \emph{adaptive price--saturation certificate}; the successful deficient voter may then depend on the actual pair $(\alpha,f)$.

\subsection{Exact eight-row classification}
\label{sec:bv-cases}
\label{sec:bv-eight}

Only the eight-row classification is needed for the main theorem: the cases with at most seven voters follow from \citet{becker2026core}. We retain our independent six- and seven-row Nash-kernel classifications in Appendix~\ref{sec:apdx:bv-predecessors}, but they are not part of the logical dependency chain of Theorem~\ref{thm:bv-main}.

The eight-row exact verification uses all three certificate types. A tight-row certificate gives a rounding directly; a fixed price--saturation certificate specifies one suitable one-deficit rounding; and an adaptive certificate allows the deficient voter, and hence the rounding, to depend on the realized pair $(\alpha,f)$. The following lemma summarizes this classification.

\begin{restatable}[\bveightclassificationtitle]{lemma}{bvclasseight}\label{lem:bv-ca-eight}
\ifbveightformal
Let $4\leq d\leq8$ and $2\leq\kappa\leq d-2$.  Let $A\in\{0,1\}^{8\times d}$ be a full-column-rank antichain with $\cD(A)\neq\varnothing$.  Let $h$ be nonimplementable and suppose $\cR_{<}(A,\kappa,h)\neq\varnothing$.  Then at least one of the following holds:
\begin{enumerate}[label=\textup{(\Roman*)}]
\item for every $f\in\cR_{<}(A,\kappa,h)$, there exists $i\in B(h)$ such that $(Af)_{i}=h_{i}$;
\item there exist $i\in B(h)$ and $z^{i}\in\{0,1\}^d$ with $\one^\top z^i=\kappa$ and $Az^i\ge h-e_i$ such that, for every $f\in\cR_{<}(A,\kappa,h)$ and $\alpha\in\cD(A)$, $G_{i}(\alpha,f)>0$, and for every coalition $S\ni i$ with $|S|\geq2$,
\begin{equation*}
  \left[
    \alpha_{j}>
    \lambda_{\kappa,|S|,h_{j}-(Az^{i})_{j}}
    \text{ for every }j\in S
  \right]
  \Longrightarrow
  P_{S,z^{i}}(\alpha,f)>0;
\end{equation*}
\item there exist a nonempty set $E\subseteq B(h)$ and vectors $z^{i}\in\{0,1\}^d$, $i\in E$, with $\one^\top z^i=\kappa$ and $Az^i\ge h-e_i$, such that the coalition implication in~\textup{(II)} holds for every $i\in E$ and $\sum_{i\in E}G_{i}(\alpha,f)>0$ for every $f\in\cR_{<}(A,\kappa,h)$ and $\alpha\in\cD(A)$.
\end{enumerate}
\else
Every admissible eight-row residual floor region satisfies at least one of the following:
\begin{enumerate}[label=\textup{(\Roman*)}]
\item it has a usable tight row;
\item it admits a fixed price--saturation certificate satisfying Theorem~\ref{thm:bv-eight-criterion}(i);
\item it admits an adaptive price--saturation certificate satisfying Theorem~\ref{thm:bv-eight-criterion}(ii).
\end{enumerate}
\fi
\end{restatable}

The fully quantified statement, computer-assisted proof, and exact census are given in Appendix~\ref{sec:apdx:bv-eight}. The eight-row census contains $13{,}269{,}620$ positive-dual matrices and $1{,}085{,}483$ regions not settled by the earlier verification steps. Of these regions, $1{,}085{,}459$ receive fixed certificates and the remaining $24$ receive adaptive certificates; no region is left unresolved.

We are now ready to prove Theorem~\ref{thm:bv-main}.

\bveightmain*

\begin{proof}
By \citet{becker2026core}, every approval-based committee election with at most seven voters has a committee in the discrete core. It therefore remains to consider an election with $n=8$ equally weighted voters.

If Assumption~\ref{assup:1} fails, Lemma~\ref{lem:bv-assumption-one-reduction} reduces the election to one with fewer than eight voters. The result of \citet{becker2026core} supplies a core committee for the residual election, and Lemma~\ref{lem:bv-assumption-one-reduction} lifts it to a core committee of the original election.

Suppose henceforth that Assumption~\ref{assup:1} holds. Delete candidates approved by no voter. If at most $k$ candidates remain, select all of them. Otherwise, Theorem~\ref{thm:frac:exist} gives a fractional Nash-core pair, and we choose one with the minimum number of fractional coordinates. If $h$ is implementable, Lemma~\ref{lem:bv-floor-rounding} gives a core committee. Otherwise, Proposition~\ref{prop:bv-reduced-matrix} gives a positive-dual full-column-rank antichain with $4\leq d\leq 8$ and $2\leq\kappa\leq d-2$.

Apply Lemma~\ref{lem:bv-ca-eight} to $(f,\alpha)$. Branch~\textup{(I)} is handled by Corollary~\ref{cor:bv-integer-coordinate}; branch~\textup{(II)} by Theorem~\ref{thm:bv-eight-criterion}(i); and branch~\textup{(III)} by Theorem~\ref{thm:bv-eight-criterion}(ii). The branches are exhaustive.
\end{proof}

The theorem is existential and does not by itself yield a polynomial-time algorithm. Its genuinely additional domain consists of eight equally weighted voters with eight distinct approval sets: an eight-voter profile with a repeated approval set has at most seven voter types and is covered by \citet{becker2026core}. Conversely, their theorem allows arbitrary weights and multiplicities across seven types, which the present theorem does not. Thus neither parameterization subsumes the other, and unrestricted discrete-core nonemptiness remains open. Although preserving the utility floors of an arbitrary supplied fractional-core allocation is NP-complete, this hardness concerns a particular rounding route rather than discrete-core nonemptiness; the formal statement, proof, and scope qualifications are given in Appendix~\ref{app:floor-hardness}.

The preceding result establishes discrete-core nonemptiness for a bounded number of voters. We now turn to a complementary approach for unrestricted elections, based on candidate-weighted harmonic welfare and an efficiently verifiable sufficient condition for discrete-core membership.

\section{Nash-Core Certificates in the General Discrete Setting}  \label{sec:disc}

This section develops complementary Nash-based certificates for unrestricted discrete elections. The discrete Nash core implies the weak Nash core, and the weak Nash core is an efficiently verifiable sufficient condition for membership in the discrete core. These certificates are not used in the bounded-voter existence proof of Section~\ref{sec:bv-core-eight}, and their existence for unrestricted elections remains conjectural.

We generalize the candidate-weighted idea in the fractional setting to the discrete setting.

Recall that the utility of voter $v$ is $u_v(x):=\sum_{c\in A_v}x_c$. Note that $u_v(x)$ is an integer in the discrete setting. The Proportional Approval Voting \citep{thiele1895om,kilgour2010approval} (PAV)-score of committee $x$ is defined as $sc_{PAV}(x):=\sum_{v\in N}H(u_v(x))$, $H(r):=\sum_{i=1}^{r}\frac{1}{i}$ is the $r$-th harmonic number. When each candidate can be selected multiple times---or equivalently, when each candidate has infinite number of copies, the committee that maximizes the PAV-score is in the core \citep{brill2024approval}. However, this approach fails when each candidate can be chosen at most once. This fact is noted in several works \citep{peters2020proportionality,peters2025core,lackner2023multi}, where the counterexample has the same idea as Example \ref{exp:nsw}, but need more voters and candidates to make everything integral.

We also introduce the weight for candidates $w_c\in (0,1]$ in the discrete setting, as we did for the fractional case, but note that the weighted utility $\phi_v(x)=\sum_{c\in A_v}w_cx_c$ may not be an integer, while the harmonic number $H(r)$ is only defined for integer $r$. We adopt the Digamma function \citep{abramowitz1948handbook}, defined as $\psi(z)=\frac{d}{dz}\ln \Gamma(z)$, where $\Gamma(z)$ is the Gamma function. A property of the Digamma function is $H(r)=\psi(r+1)+\gamma$ for each non-negative integer $r$, where $\gamma$ is the Euler-Mascheroni constant, so we directly define the continuous generalization of the harmonic number as $H(z)=\psi(z+1)+\gamma$ for any $z\geq 0$. Then $H(z)$ is concave and we have $H(z+1)=H(z)+\frac{1}{z+1}$ for any $z\geq 0$. We summarize these properties as the following lemma.

\begin{lemma}[\citep{abramowitz1948handbook}]    \label{lem:cont:h}
    Let $H(z)=\frac{d}{dz}\ln \Gamma(z+1)+\gamma$, where $\Gamma(\cdot)$ is the Gamma function and $\gamma$ is the Euler–Mascheroni constant. Then
    \begin{itemize}
        \item When $z\geq 0$ is an integer, $H(z)=\sum_{i=1}^{z}\frac{1}{i}$, matching the harmonic number.
        \item For $z\geq 1$, $H(z)-H(z-1)=\frac{1}{z}$.
        \item $H(z)$ is concave when $z\geq 0$.
    \end{itemize}
\end{lemma}

A direct corollary of the second and third properties in this lemma is following.
\begin{corollary} \label{cor:cont:h}
We have
    \begin{itemize}
        \item For $a\in (0,1]$, $z\geq a$, $H(z)-H(z-a)\leq \frac{a}{z}$,
        \item For $a\in [0,1]$, $z\geq 0$, $H(z+a)-H(z)\geq \frac{a}{z+1}$.
    \end{itemize}
\end{corollary}

We define the continuous PAV score as
\begin{definition}[candidate-weighted PAV score]
    $sc_{PAV}(w,x):=\sum_{v\in N}H(\phi_v(x))=\sum_{v\in N}H(\sum_{c\in A_v}w_cx_c)$.
\end{definition}

We define the discrete Nash core analogously to the fractional Nash core in Definition~\ref{def:nashcore}.
\begin{definition}[discrete Nash core]  \label{def:disccore}
    For $w\in (0,1]^m$ and $x\in \{0,1\}^m$ where $\|x\|_1\leq k$, we say $(w,x)$ is in the discrete Nash core, if both of the following hold:
    \begin{itemize}
        \item For each candidate $c\in C$, $w_c=1$ if $x_c=0$.
        \item $sc_{PAV}(w,x)=\max\limits_{z\in \{0,1,\dots,k\}^m: \|z\|_1\leq k}sc_{PAV}(w,z)$.
    \end{itemize}
\end{definition}

We also have the following definition.
\begin{definition}[weak Nash core]
    For $w\in (0,1]^m$ and $x\in \{0,1\}^m$ where $\|x\|_1\leq k$, we define the weighted utility of voter $v$ as $\phi_v(x)=\sum_{c\in A_v}w_cx_c$. We say $(w,x)$ is in the weak Nash core if both of the following hold:
    \begin{itemize}
        \item For each candidate $c\in C$, $w_c=1$ if $x_c=0$.
        \item For each candidate $c\in C$, $\sum_{v:c\in A_v}\frac{w_c}{\phi_v(x)+1}< \frac{n}{k}$.
    \end{itemize}
\end{definition}

In this definition, the term $\sum_{v:c\in A_v}\frac{w_c}{\phi_v(x)+1}$ in the second condition matches the idea of \emph{marginal utility} in \citep{aziz2017justified}.

We show discrete Nash core solutions are in the weak Nash core, and weak Nash core solutions are in the discrete core. We defer the proof of this theorem to Appendix~\ref{sec:proof:disc:incore}.

\discincore*

\section{Heuristically Computing the Core}  \label{sec:experi}

We develop heuristic algorithms and evaluate them on datasets from the Pabulib repository \citep{faliszewski2023participatory}, focusing primarily on the 100 largest datasets, each with over 3000 voters and 50 candidates, while also considering an additional 1000 smaller datasets. The ten largest datasets contain approximately $10^5$ voters and 100 candidates. Our algorithms approximately compute both the fractional Nash core solution and the weak Nash core solution, the latter of which also satisfies the discrete core condition. All experiments are conducted on a personal computer. For the vast majority of datasets, the runtime is under 1 minute, while the largest instances require between 5 and 10 minutes. We discuss more details in Appendix \ref{sec:apdx:exp}.

\section{Conclusions and Future Directions}

We study the core in approval voting through fractional and discrete approaches. In the discrete setting, we prove an exact core-existence theorem for elections with at most eight equally weighted voters and develop complementary Nash-based certificates for unrestricted elections.

In the fractional setting, we introduce the fractional Nash core---a candidate-weighted generalization of Nash social welfare maximization---and propose a proportional payment scheme (reminiscent of the Method of Equal Shares) that incorporates both prices and payments. We show that these two notions coincide, that they imply the fractional core, and that such a solution exists under Assumption~\ref{assup:1}.

In the discrete setting, we show that every approval-based committee election with at most eight individual equally weighted voters has a core committee. The proof uses the fractional Nash-core structure and an exact computer-assisted classification of eight-row residual regions. For unrestricted elections, we introduce the discrete Nash core---a candidate-weighted generalization of Proportional Approval Voting---and define the weak Nash core. We show that the discrete Nash core implies the weak Nash core, and that the weak Nash core implies the discrete core. Notably, the weak Nash core is straightforward to verify.

Additionally, we propose heuristic searches for candidate fractional Nash-core and weak Nash-core pairs. Although these methods provide numerical evidence only and do not prove convergence or existence, they run efficiently on real-world voting data.

We propose the following future directions.
\begin{itemize}
    \item \textbf{Uniqueness of fractional Nash core.} Although Nash core solution is not always unique, we conjecture that, for any Nash core solution $(w,x)$, the voter utility $\sum_{c\in A_v}x_c$ remains the same for each $v\in N$. In other words, every voter's utility is invariant under different Nash core solutions.
    \item \textbf{Efficient computation of fractional Nash core.} Although our heuristic search for candidate fractional Nash-core pairs is empirically efficient, it lacks a formal theoretical guarantee. Does there exist an efficient algorithm with a theoretical guarantee for computing the fractional Nash core?
    \item \textbf{Unrestricted discrete-core nonemptiness.} We establish nonemptiness for elections with at most eight equally weighted voters. Whether every unrestricted approval-based committee election admits a discrete-core committee remains open.
    \item \textbf{Existence of the weak Nash core solution.} Verifying whether a solution belongs to the discrete core is coNP-hard \citep{brill2024approval}. However, determining whether a solution is in the weak Nash core is straightforward. Since every weak Nash core solution lies within the discrete core---and our heuristic finds candidate pairs on real voting data---demonstrating existence for unrestricted elections offers a potential pathway beyond the at-most-eight-voter theorem.
\end{itemize}

\backmatter

\bmhead{Code availability}

The code is available at \url{https://github.com/nashcore-code/code}. The \texttt{heuristics} directory contains the heuristic experiments, and the \texttt{discrete-core-nonemptiness} directory contains the computer-assisted proof of discrete-core nonemptiness.

\begin{appendices}

\section{Omitted Proofs} \label{sec:apdx:proof}

\subsection{Proof of Theorem \ref{thm:frac:equiv}} \label{sec:proof:frac:equiv}

\fracequiv*

We split the argument into two lemmas, each of which handles one direction of implication.

\begin{lemma}[Fractional Nash Core \(\Longrightarrow\) Proportional Payment]
\label{lem:frac:equiv:forward}
Suppose \((w,x)\) is in the fractional Nash core (Definition~\ref{def:nashcore}). Then \((w,x)\) is supported by a proportional payment (Definition~\ref{def:payment}).
\end{lemma}
\begin{proof}
Since \((w,x)\) is in the fractional Nash core, we have:
\begin{enumerate}
    \item Weight-Selection Coupling: If \(w_c < 1\), then \(x_c=1\).
    \item Nash Optimality: The vector \(x\) solves
    \[
      \max_{z \ge 0,\, \|z\|_1\le k}\;\; \sum_{v\in N} \log\Bigl(\sum_{c\in A_v} w_c z_c\Bigr).
    \]
\end{enumerate}

By the standing assumption that $A_v\neq\emptyset$ for every voter and because $w_c>0$, the feasible vector with $z_c=k/m$ for every $c\in C$ gives every voter positive weighted utility. Hence the Nash-welfare objective has a finite feasible value. Since $x$ is optimal, its objective value is finite, and therefore $\phi_v(x)>0$ for every voter $v$.

Moreover, $\|x\|_1=k$. Otherwise, since $\phi_v(x)>0$ for every voter, scaling $x$ by $k/\|x\|_1>1$ would remain feasible and strictly increase every weighted voter utility, contradicting Nash optimality.

By standard convex optimization (via KKT conditions), there is a Lagrange multiplier \(\lambda\) for the constraint \(\sum_{c}z_c \le k\). For each candidate \(c\) with \(x_c>0\), we have
\[
  \sum_{v : c\in A_v} \frac{w_c}{\sum_{c'\in A_v}w_{c'}x_{c'}}
  \;=\;\lambda.
\]
For each candidate $c$ with $x_c=0$, the KKT conditions also give
\begin{equation*}
  \sum_{v:c\in A_v}\frac{w_c}{\phi_v(x)}
  \leq \lambda.
\end{equation*}
Rewriting \(\sum_{c'\in A_v} w_{c'} x_{c'}\) as \(\phi_v(x)\), and multiplying both sides by \(x_c\), we get
\[
  \sum_{v : c\in A_v} \frac{w_c x_c}{\phi_v(x)}
  \;=\;\lambda\,x_c.
\]
Since \(\|x\|_1 = k\) at optimum, we have
\[
\lambda k
= \sum_{c\in C}\lambda x_c
= \sum_{c\in C} \sum_{v : c\in A_v} \frac{w_c x_c}{\phi_v(x)}
= \sum_{v\in N} \sum_{c\in A_v} \frac{w_c x_c}{\phi_v(x)}
= n
,
\]
identify \(\lambda = \frac{n}{k}\). This gives
\begin{align}   \label{eq:deriv}
  \sum_{v : c\in A_v} \frac{w_c x_c}{\phi_v(x)}
  \;=\;\frac{n}{k} x_c.
\end{align}
In Definition \ref{def:payment}, we define
\[
  p_{vc} \;=\;
  \begin{cases}
     \dfrac{w_c x_c}{\phi_v(x)}, & \text{if } c\in A_v,\\
     0, & \text{otherwise}.
  \end{cases}
\]
Then
\[
  \sum_{v\in N} p_{vc}
  \;=\;\sum_{v : c\in A_v} \frac{w_c x_c}{\phi_v(x)}
  \;=\;\frac{n}{k}\,x_c,
\]
which is precisely the \emph{affordability} condition from Definition~\ref{def:payment}.

The \emph{stability} condition follows from complementary slackness: if \(x_c<1\), then by the fractional Nash core property, \(w_c=1\). Hence
\begin{align}   \label{ieq:deriv}
  \sum_{v : c\in A_v} \frac{1}{\phi_v(x)}
  \;\le\;\frac{n}{k},
\end{align}
matching the definition.

Finally, the \emph{weight-selection coupling} is already given: \(w_c<1\implies x_c=1\).

Hence \((w,x)\) satisfies all conditions of Definition~\ref{def:payment}.
\end{proof}

\begin{lemma}[Proportional Payment \(\Longrightarrow\) Fractional Nash Core]
\label{lem:frac:equiv:backward}
Suppose \((w,x)\) is supported by a proportional payment (Definition~\ref{def:payment}). Then \((w,x)\) is in the fractional Nash core (Definition~\ref{def:nashcore}).
\end{lemma}
\begin{proof}
Given a proportional payment scheme, we have:
\begin{itemize}
    \item Positive utility: \(\phi_v(x) > 0\) for all \(v\).
    \item Affordability:
    \[
      \sum_{v \in N} p_{vc}
      \;=\;\frac{n}{k}\,x_c
      \quad\text{and}\quad
      p_{vc} \;=\; \frac{w_c x_c}{\phi_v(x)}
      \quad\text{if }c\in A_v.
    \]
    Thus
    \[
      \sum_{v : c\in A_v} p_{vc}
      \;=\; w_c x_c \sum_{v: c\in A_v} \frac{1}{\phi_v(x)}
      \;=\;\frac{n}{k}\,x_c.
    \]
    If \(x_c>0\), we get \(\sum_{v: c\in A_v} \frac{1}{\phi_v(x)} = \frac{n}{k w_c}\).

    \item Stability: For \(x_c<1\), the condition in Definition~\ref{def:payment} ensures \(w_c=1\), and
    \[
      \sum_{v : c\in A_v} \frac{1}{\phi_v(x)}
      \;\le\;\frac{n}{k}.
    \]

    \item Weight-Selection Coupling: If \(w_c<1\), then \(x_c=1\), per the definition.
\end{itemize}
Summing affordability over all candidates gives
\begin{equation*}
  \frac{n}{k}\sum_{c\in C}x_c
  =\sum_{c\in C}\sum_{v\in N}p_{vc}
  =\sum_{v\in N}1
  =n,
\end{equation*}
and hence $\|x\|_1=k$.

Thus, with $\lambda=n/k$,
\begin{equation*}
  \sum_{v:c\in A_v}\frac{w_c}{\phi_v(x)}
  \leq \lambda
  \quad\text{for every }c,
\end{equation*}
with equality whenever $x_c>0$. Together with $\|x\|_1=k$, these are precisely the KKT conditions for the candidate-weighted Nash social welfare maximization. Since its objective is concave and its feasible region is convex, these conditions are sufficient for optimality. Hence \(x\) must be an optimal solution to
\[
  \max_{z \ge 0,\, \|z\|_1 \le k} \sum_{v\in N} \log\Bigl(\sum_{c\in A_v} w_c z_c\Bigr).
\]
Combined with the condition that \(w_c<1\implies x_c=1\), we conclude \((w,x)\) is in the fractional Nash core.
\end{proof}

The proof of Theorem \ref{thm:frac:equiv} follows immediately from Lemmas \ref{lem:frac:equiv:forward} and \ref{lem:frac:equiv:backward}.

\subsection{Proof of Theorem \ref{thm:lindahl}} \label{sec:proof:lindahl}

\begin{proof}

Recall that we defined the payment $p_{vc}$ in Definition \ref{def:payment} as following.
\[
  p_{vc} \;=\;
  \begin{cases}
     \dfrac{w_c x_c}{\phi_v(x)}, & \text{if } c\in A_v,\\
     0, & \text{otherwise},
  \end{cases}
\]
where $\phi_v(x) = \sum_{c\in A_v} w_{c} x_{c}$.

For each voter $v\in N$, We define the price vector $q_v$ as the payment per unit, which is $p_{v}$ when $x=1$. Additionally, we apply a scaling factor of $\frac{k}{n}$.
\[
  q_{vc} \;=\;
  \begin{cases}
     \dfrac{k}{n}\cdot \dfrac{w_c}{\phi_v(x)}, & \text{if } c\in A_v,\\
     0, & \text{otherwise}.
  \end{cases}
\]

Now we show $(q,x)$ constitutes a Lindahl equilibrium.

\textbf{Condition (1).}
Fix a voter $v\in N$. Note that
\[
  q_v^\top x
  \;=\;
  \sum_{c \in A_v} \frac{k}{n} \frac{w_c}{\phi_v(x)}\,x_c
  \;=\;
  \frac{k}{n} \frac{1}{\phi_v(x)} \sum_{c\in A_v} w_c \, x_c
  \;=\;
  \frac{k}{n},
\]
so $y_v = x$ satisfies the constraint $q_v^{T} y_v\leq \frac{k}{n}$. Note that $q_{vc}$ is proportional to $w_c$ for $c\in A_v$, and $x_c=1$ if $w_c<1$. So for another solution $y\in [0,1]^m$, if $\sum_{c\in A_v}y_{vc} > \sum_{c\in A_v}x_c$, then we have
\begin{align*}
    q_v^{T} y_v - q_v^{T} x
    = & \sum_{c\in A_v: w_c=1} \frac{k}{n} \frac{1}{\phi_v(x)} \left(y_{vc} - x_{c}\right) + \sum_{c\in A_v: w_c<1} \frac{k}{n} \frac{w_c}{\phi_v(x)} \left(y_{vc} - 1\right)   \\
    \geq & \sum_{c\in A_v: w_c=1} \frac{k}{n} \frac{1}{\phi_v(x)} \left(y_{vc} - x_{c}\right) + \sum_{c\in A_v: w_c<1} \frac{k}{n} \frac{1}{\phi_v(x)} \left(y_{vc} - 1\right)   \\
    = & \frac{k}{n \phi_v(x)} \left( \sum_{c\in A_v}y_{vc} - \sum_{c\in A_v}x_c \right)
    >0,
\end{align*}
implying $q_v^{T} y_v > q_v^{T} x = \frac{k}{n}$, so $y$ is not feasible. Thus $x$ is a utility-maximizing choice for voter $v$.

\textbf{Condition (2).}
Let $Q = \sum_{v\in N} q_v$, and define the aggregator's profit function:
\[
  \Pi(z) \;=\; Q^\top z \;-\; \|z\|_1
  \quad
  \text{for } z \in [0,1]^m,\, \|z\|_1 \le k,
\]
as required by Definition~\ref{def:lindahl}. We show that $z=x$ maximizes $\Pi(z)$. Write $Q_c=\sum_{v\in N}q_{vc}$. First, suppose that $x_c>0$. The affordability condition gives
\[
  x_c w_c\sum_{v:c\in A_v}\frac{1}{\phi_v(x)}
  =\sum_{v\in N}p_{vc}
  =\frac{n}{k}x_c.
\]
Dividing by $x_c>0$ and using the definition of $q_{vc}$ yields
\[
  Q_c
  =\frac{k}{n}w_c\sum_{v:c\in A_v}\frac{1}{\phi_v(x)}
  =1.
\]
Next, suppose that $x_c=0$. The stability condition applies because $x_c<1$, and weight-selection coupling gives $w_c=1$. Therefore,
\[
  Q_c
  =\frac{k}{n}\sum_{v:c\in A_v}\frac{w_c}{\phi_v(x)}
  =\frac{k}{n}\sum_{v:c\in A_v}\frac{1}{\phi_v(x)}
  \leq 1.
\]
Thus $Q_c\leq1$ for every $c\in C$, with equality whenever $x_c>0$. Consequently, for every $z\in[0,1]^m$ satisfying $\|z\|_1\leq k$,
\[
  \Pi(z)
  =\sum_{c\in C}(Q_c-1)z_c
  \leq 0
  =\sum_{c\in C}(Q_c-1)x_c
  =\Pi(x).
\]
Hence $x$ maximizes the aggregator's profit over the feasible set in Definition~\ref{def:lindahl}, establishing Condition~(2).

Since both conditions of Definition \ref{def:lindahl} are verified, $(q_1,\dots,q_n,x)$ is indeed a Lindahl equilibrium.

\end{proof}

\subsection{Proof of Theorem \ref{thm:disc:incore}} \label{sec:proof:disc:incore}

\discincore*

We split this theorem into two lemmas.

\begin{lemma}   \label{lem:disc:incore:1}
    Suppose $w\in (0,1]^m$ and $x\in \{0,1\}^m$ where $\|x\|_1\leq k$. If $(w,x)$ is in the discrete Nash core, then it is in the weak Nash core.
\end{lemma}

\begin{proof}

Assume $(w,x)$ is in the discrete Nash core but it is not in the weak Nash core. Then there exists $c\in C$, such that $\sum_{v:c\in A_v}\frac{w_c}{\phi_v(x)+1}\geq \frac{n}{k}$.

Let $y=x+e_c$, where \(e_c\) is the standard basis vector having a 1 in the \(c\)-th position and 0 in every other position. Then we have
\begin{align*}
    sc_{\text{PAV}}(w,y) - sc_{\text{PAV}}(w,x)
    = & \sum_{v:c\in A_v} \left( H(\phi_v(y))-H(\phi_v(x)) \right) \\
    \geq & \sum_{v:c\in A_v} \frac{w_c}{\phi_v(x)+1}
    \geq \frac{n}{k},
\end{align*}
where the first inequality follows from Corollary \ref{cor:cont:h}.

Then we consider drawing $c'$ with probability proportional to $y_{c'}$, and define $z=y-e_{c'}$. If $\|x\|_1<k$, then $y$ is feasible for the maximization in Definition~\ref{def:disccore}, and the preceding inequality contradicts optimality. Hence we may assume $\|x\|_1=k$, so $\|y\|_1=k+1$. Then we have
\begin{align*}
    & \E{sc_{\text{PAV}}(w,z) - sc_{\text{PAV}}(w,y)} \\
    =& \frac{1}{k+1}\sum_{c': y_{c'}>0} y_{c'} \sum_{v:c'\in A_v} \left( H(\phi_v(y)-w_{c'})-H(\phi_v(y)) \right)    \\
    \geq & - \frac{1}{k+1}\sum_{c': y_{c'}>0} y_{c'} \sum_{v:c'\in A_v} \frac{w_{c'}}{\phi_v(y)}    \\
    = & - \frac{1}{k+1}\sum_{v:\phi_v(y)>0} \frac{\sum_{c'\in A_v}w_{c'}y_{c'}}{\phi_v(y)} \\
    = & -\frac{|\{v\in N:\phi_v(y)>0\}|}{k+1}
    \geq -\frac{n}{k+1},
\end{align*}
where the inequality follows from Corollary \ref{cor:cont:h}. It implies there exists $c'\in C$ such that $y_{c'}\geq 1$ and $sc_{PAV}(w,y-e_{c'}) - sc_{PAV}(w,y)\geq -\frac{n}{k+1}$. Combining with $sc_{PAV}(w,y) - sc_{PAV}(w,x)\geq \frac{n}{k}$, we have $sc_{PAV}(w,y-e_{c'})>sc_{PAV}(w,x)$. Note that $\|y-e_{c'}\|=k$, so it violates the optimality condition $sc_{PAV}(w,x)=\max\limits_{z\in \{0,1,\dots,k\}^m: \|z\|_1\leq k}sc_{PAV}(w,z)$, which contradicts to the assumption that $(w,x)$ is in the discrete Nash core.

As a result, $(w,x)$ is in the weak Nash core.

\end{proof}

\begin{lemma}   \label{lem:disc:incore:2}
    Suppose $w\in (0,1]^m$ and $x\in \{0,1\}^m$ where $\|x\|_1\leq k$. If $(w,x)$ is in the weak Nash core, then $x$ is in the discrete core.
\end{lemma}

\begin{proof}
  Assume $(w,x)$ is in the weak Nash core, but $x$ is not a discrete core solution. Then there exists $S\subseteq N$ and nonzero $y\in \{0,1\}^{m}$ such that $\onenorm{y}\leq \frac{k|S|}{n}$, and $1+\sum_{c\in A_v}x_c\leq \sum_{c\in A_v}y_c$ for all $v\in S$.

  Let $W = \{c\in C: w_c<1\}$, so $x_c=1$ for $c\in W$. For $v\in S$, since $1+\sum_{c\in A_v}x_c\leq \sum_{c\in A_v}y_c$, we have
  \begin{align*}
     \sum_{c\in A_v}w_{c}y_{c} - \sum_{c\in A_v}w_{c}x_{c}
    =& \sum_{c\in A_v\cap W} w_{c} (y_{c} - 1) + \sum_{c\in A_v\setminus W} (y_{c}-x_{c})   \\
    \geq & \sum_{c\in A_v\cap W} (y_{c} - 1) + \sum_{c\in A_v\setminus W} (y_{c}-x_{c})   \\
    =& \sum_{c\in A_v}y_c - \sum_{c\in A_v}x_c \geq 1,
  \end{align*}
  implying $1+\sum_{c\in A_v}w_{c}x_{c} \leq \sum_{c\in A_v}w_{c}y_{c}$.

  Recall that since $(w,x)$ is in the weak Nash core, for each candidate $c\in C$, $\sum_{v:c\in A_v}\frac{w_c}{\phi_v(x)+1}< \frac{n}{k}$, where $\phi_v(x)=\sum_{c\in A_v}w_cx_c$. Then we have
  \begin{align*}
      \frac{n}{k}\|y\|_1
      >& \sum_{c\in C}y_c \sum_{v:c\in A_v}\frac{w_c}{\phi_v(x)+1}  \\
      = & \sum_{v\in N} \sum_{c\in A_v}\frac{w_c y_c}{\phi_v(x)+1}  \\
      \geq & \sum_{v\in S} \sum_{c\in A_v}\frac{w_c y_c}{\phi_v(x)+1}  \\
      = & \sum_{v\in S}\frac{\phi_v(y)}{\phi_v(x)+1}
      \geq |S| \geq \frac{n}{k}\|y\|_1,
  \end{align*}
  which is a contradiction.

  As a result, $x$ is a discrete core solution.
\end{proof}

The proof of Theorem \ref{thm:disc:incore} follows immediately from Lemmas \ref{lem:disc:incore:1} and \ref{lem:disc:incore:2}. A direct corollary of Theorem \ref{thm:disc:incore} is the following.
\begin{corollary}
    Suppose $w\in (0,1]^m$ and $x\in \{0,1\}^m$ where $\|x\|_1\leq k$. If $(w,x)$ is in the discrete Nash core, then $x$ is in the discrete core.
\end{corollary}

\section{Supporting Results for Theorem~\ref{thm:bv-main}}
\label{sec:apdx:bv-analytical}
\label{sec:apdx:computer-assisted}

This appendix provides the supporting results used in the proof of Theorem~\ref{thm:bv-main}. We first prove the minimum-fractionality reduction in Proposition~\ref{prop:bv-reduced-matrix}, derive the blocking-coalition inequalities in Proposition~\ref{prop:bv-blocker-conditions}, and establish the eight-voter saturation bound in Lemma~\ref{lem:bv-lambda}. We then give the computer-assisted proof of the eight-row classification in Lemma~\ref{lem:bv-ca-eight}. The reductions used by the exact checker are proved in Proposition~\ref{prop:bv-polyhedral}, and the completeness and soundness of the enumeration and verification are established in Proposition~\ref{prop:bv-ca-verification}. We retain our independent six- and seven-row classifications in Section~\ref{sec:apdx:bv-predecessors} as supplementary results; they are not used to prove Theorem~\ref{thm:bv-main}. Finally, Proposition~\ref{prop:bv-floor-rounding-hardness} proves that preserving all voter-utility floors when rounding a specified fractional committee is NP-complete.

\subsection{Proof of the minimum-fractionality reduction}
\label{sec:proof:bv-reduced-matrix}

\bvreducedmatrix*

\begin{proof}[Proof of Proposition~\ref{prop:bv-reduced-matrix}]
If two fractional candidates have identical supports, transfer mass from one to the other while preserving their sum until one coordinate reaches zero or one.  Every weighted voter utility and the total budget remain unchanged, so the resulting pair is still a fractional Nash-core pair but has fewer fractional coordinates.

If two fractional candidates $a,b$ satisfy $N(a)\subsetneq N(b)$, then for sufficiently small $\varepsilon>0$ the changes $f_a\leftarrow f_a-\varepsilon$ and $f_b\leftarrow f_b+\varepsilon$ preserve the budget, weakly increase every weighted voter utility, and strictly increase the utility of each voter in $N(b)\setminus N(a)$.  This contradicts Nash optimality.  This proves part~\textup{(i)}.

For part~\textup{(ii)}, suppose $A\gamma=0$ for some nonzero $\gamma\in\R^d$, oriented so that $\one^\top\gamma\le0$.  Since $f\in(0,1)^d$, move from $f$ in direction $\gamma$ until the first coordinate reaches zero or one.  All voter utilities remain unchanged.  If $\one^\top\gamma=0$, the number of fractional coordinates decreases.  If $\one^\top\gamma<0$, the resulting pair is still Nash-optimal but uses less than the full budget. Scaling it to use the full budget would strictly increase every weighted utility, contradicting the Nash optimality condition in Definition~\ref{def:nashcore}. Thus $A$ has full column rank.

For part~\textup{(iii)}, the cases $\kappa=0$ and $\kappa=d$ are already integral.  If $\kappa=1$, every row with positive floor must contain every fractional column, so any one column implements all positive floor coordinates.  If $\kappa=d-1$, put $q:=\one-f$.  Then $q>0$ and $\one^\top q=1$.  A row of degree zero has floor zero; a row of degree $1\le d_i<d$ has $(Af)_i=d_i-(Aq)_i\in(d_i-1,d_i)$ and hence floor $d_i-1$; and a row of degree $d$ has floor $d-1$.  Selecting all but any one residual column therefore implements every row floor.  Nonimplementability forces $2\le\kappa\le d-2$.

Finally, every fractional candidate has weight one, so \eqref{eq:bv-price-by-status} gives $\alpha(N(c))=1$ for every column of $A$. These equations are exactly $A^\top\alpha=\one$, with $\alpha>0$.
\end{proof}

\subsection{Proof of the blocking-coalition inequalities}
\label{sec:proof:bv-blocker-conditions}

\bvblockerconditions*

\begin{proof}[Proof of Proposition~\ref{prop:bv-blocker-conditions}]
\medskip\noindent\emph{(i).} Put $R:=T\setminus I$ and $D:=I\setminus T$. Strict integral improvement gives
\begin{equation}\label{eq:bv-improvement-RD}
 u_i(R)-u_i(D)\ge q_i+1\quad\forall i\in S.
\end{equation}
Every candidate in $R$ lies in $F\cup O$, so Equation~\eqref{eq:bv-nash-price} implies
\begin{equation}\label{eq:bv-R-price}
 \sum_{i\in S}\alpha_i u_i(R)\le |R|.
\end{equation}
The total normalized payment made by $S$ to the fully selected candidates is $B_S:=\sum_{i\in S}\alpha_i\beta_i =t|S|-\sum_{i\in S}\alpha_i r_i$. The retained full candidates $I\setminus D$ carry at most $|I\setminus D|$ of this payment.  Since $w_c\le1$, ordinary $\alpha$-mass dominates weighted payment mass, and therefore
\begin{equation}\label{eq:bv-D-lower}
 \sum_{i\in S}\alpha_i u_i(D)
 \ge\max\{0,B_S-|I\setminus D|\}.
\end{equation}
If $B_S\le |I\setminus D|$, then $|R|\le t|S|-B_S$ because $|I\setminus D|+|R|=|T|\le t|S|$. If $B_S>|I\setminus D|$, equations~\eqref{eq:bv-R-price} and \eqref{eq:bv-D-lower} give
\[
 \sum_{i\in S}\alpha_i\bigl(u_i(R)-u_i(D)\bigr)
 \le |R|+|I\setminus D|-B_S
 \le t|S|-B_S.
\]
Thus in both cases
\[
 \sum_{i\in S}\alpha_i\bigl(u_i(R)-u_i(D)\bigr)
 \le t|S|-B_S
 =\sum_{i\in S}\alpha_i r_i.
\]
Comparison with \eqref{eq:bv-improvement-RD} proves part~\textup{(i)}.

\medskip\noindent\emph{(ii).} By \eqref{eq:bv-beta-relation}, the integer utility from the fully selected candidates satisfies $g_i\ge\ceil{t/\alpha_i-r_i}$.  The incumbent committee gives voter $i$ utility $g_i+q_i$.  A strict integral improvement therefore gives at least $\ceil{t/\alpha_i-r_i}+q_i+1 =\ceil{t/\alpha_i+q_i-r_i}+1$ approved candidates.  This cannot exceed the cardinality of the deviation, which is at most $\floor{ts}$, proving part~\textup{(ii)}.
\end{proof}

\subsection{Proof of the eight-voter saturation bound}
\label{sec:proof:bv-lambda}

\bvlambda*

\begin{proof}[Proof of Lemma~\ref{lem:bv-lambda}]
Put $a:=h_i-(Az)_i$ and $L:=\floor{ks/8}$.  Since $r_i=h_i+\delta_i$ with $\delta_i<1$, Proposition~\ref{prop:bv-blocker-conditions}(ii) gives
\[
  \frac{k}{8\alpha_i}-a-\delta_i\le L-1,
\]
and hence $k/(8\alpha_i)<L+a$.  Thus $L+a>0$ and
\[
  \alpha_i>\frac{k}{8(\floor{ks/8}+a)}.
\]
It remains to minimize the right-hand side over all integers $k\ge\kappa$.

Fix a residue class modulo eight and let $B_r\in\{\kappa,\ldots,\kappa+7\}$ be its first representative at least $\kappa$.  Write $k=B_r+8\ell$ and $D_r:=\floor{B_rs/8}+a$.  On this residue class the ratio is
\[
  g_r(\ell):=\frac{B_r+8\ell}{8(D_r+\ell s)}.
\]
Its successive differences have constant sign, namely the sign of $8D_r-sB_r=8a-(B_rs\bmod 8)$, and $g_r(\ell)\to1/s$ as $\ell\to\infty$.  If $D_r>0$, the infimum on the class is either $g_r(0)$ or $1/s$.  If $D_r\le0$, then $a\le0$.  The displayed sign is strictly negative: this is immediate when $a<0$, while for $a=0$ the condition $D_r=0$ forces $0<B_rs<8$, so $(B_rs\bmod 8)>0$.  Hence the admissible tail is decreasing and its infimum is $1/s$.  Taking the minimum over the eight residue classes gives \eqref{eq:bv-lambda-def} and proves the lemma.
\end{proof}

\subsection{Exact Eight-Row Classification}

This section contains the computer-assisted proof of Lemma~\ref{lem:bv-ca-eight}, together with its exact census table. The mathematical justification for the polyhedral reductions, enumeration completeness, and checker soundness is included here. Thus the proof-relevant eight-row verification is self-contained within the article appendices. The source code for the exact enumeration and verification is available at \url{https://github.com/nashcore-code/code/tree/main/discrete-core-nonemptiness}. The released eight-row verification driver consumes only the data and checkers in the modules \texttt{n8/m4} through \texttt{n8/m8}; it does not import the accepted six- or seven-row classification files.

\subsubsection{Exact verification framework}
\label{sec:bv-verification-overview}

The eight-row classification is a finite universal statement over residual matrices, floor regions, and Nash-price domains. Its verification uses the following elementary polyhedral facts. The same framework is also used for the independent six- and seven-row classifications retained later in Section~\ref{sec:apdx:bv-predecessors}.

\begin{proposition}[Polyhedral reductions used by the exact checkers]
\label{prop:bv-polyhedral}
Assume $\cR_{<}(A,\kappa,h)\ne\varnothing$.
\begin{enumerate}[label=\textup{(\roman*)}]
\item The open region is dense in its closed relaxation: $\cl(\cR_{<}(A,\kappa,h))=\cR_{\leq}(A,\kappa,h)$.  Consequently, every continuous affine function $L$ satisfies
\[
 \inf_{f\in\cR_{<}(A,\kappa,h)}L(f)
 =\min_{f\in\cR_{\leq}(A,\kappa,h)}L(f).
\]
\item If $Q$ is a nonempty compact convex set and $L(\alpha,f)$ is affine in $\alpha$ for every fixed $f\in Q$, then $g(\alpha):=\min_{f\in Q}L(\alpha,f)$ is concave.
\item If $A\ge0$, $v\ge0$, and $A^\top v=0$, then $v$ is supported only on zero rows of $A$.  Along such recession directions, one-deficit singleton margins and adaptive sums are constant, while coalition price margins are nondecreasing.
\item If $P$ and $Q$ are nonempty bounded polytopes and $B(x,y)$ is affine in each variable separately, then the minimum of $B$ on $P\times Q$ is attained at a pair of vertices.
\end{enumerate}
\end{proposition}

\begin{proof}
For part~\textup{(i)}, fix $f^\circ\in\cR_{<}(A,\kappa,h)$.  For any $\bar f\in\cR_{\leq}(A,\kappa,h)$ and $\varepsilon\in(0,1]$, the point $(1-\varepsilon)\bar f+\varepsilon f^\circ$ belongs to the open region and converges to $\bar f$.  Continuity and compactness give the equality of optima.

For part~\textup{(ii)}, if $0\le\theta\le1$, then
\begin{align*}
 g(\theta\alpha+(1-\theta)\alpha')
 &=\min_{f\in Q}\bigl(\theta L(\alpha,f)
 +(1-\theta)L(\alpha',f)\bigr)\\
 &\ge\theta g(\alpha)+(1-\theta)g(\alpha').
\end{align*}
For part~\textup{(iii)}, each column $c$ satisfies $0=(A^\top v)_c=\sum_i A_{ic}v_i$.  All summands are nonnegative, so $v_i>0$ is possible only on a zero row.  On such a row $(Az)_i=(Af)_i=0$, which gives the asserted slopes of the relevant margins. For part~\textup{(iv)}, start with a minimizing pair.  Holding one coordinate fixed, replace the other by a minimizing vertex, and then repeat for the second coordinate.
\end{proof}

Part~\textup{(i)} permits affine checks on the closed floor polytope, provided that a zero boundary minimum is accepted only when a separate relative-interior argument supplies strictness.  Part~\textup{(ii)} reduces a one-dimensional dual domain to its endpoints and recession behavior, while part~\textup{(iv)} reduces bounded two-dimensional price--floor products to vertex pairs.  Full column rank leaves a unique price vector for $d=8$, a one-dimensional domain for $d=7$, and a two-dimensional domain for $d=6$; the remaining recession directions are covered by part~\textup{(iii)}.

\begin{proposition}[Soundness and completeness of the finite verification]
\label{prop:bv-ca-verification}
The exact-verification pipeline has the following properties:
\begin{enumerate}[label=\textup{(\roman*)}]
\item up to voter and candidate relabeling, every admissible residual matrix is enumerated;
\item every nonimplementable floor with a nonempty open region is examined unless an exact infeasibility certificate removes it;
\item every subset or cover-dual rejection is sound;
\item every retained record satisfies the relevant classification branch on its complete real floor and price domains;
\item every unresolved eight-row region occurs exactly once in the accepted certificate file or the failure file, and the failure files are empty.
\end{enumerate}
\end{proposition}

\begin{proof}
For part~\textup{(i)}, candidate supports are encoded as bit masks and added by canonical augmentation.  Deleting a column from an antichain produces a parent antichain; when positive duality and full column rank are imposed, the same deletion preserves both properties.  Induction on the number of columns therefore reaches every isomorphism class, while exact canonicalization removes only duplicate children.

For part~\textup{(ii)}, implementable floors form the down-closure of the finite set $\cU(A,\kappa)$.  Processing these utility vectors one at a time and retaining the coordinatewise-minimal vectors not dominated by those already processed maintains exactly the minimal nonimplementable floors. Enumerating their upper cones inside the degree-bounded box therefore examines every nonimplementable floor.  A common-slack rational linear program decides whether its open region is nonempty.

For part~\textup{(iii)}, the subset filter compares
$\sum_{i\in S}h_i$ with the exact maximum obtained by summing the $\kappa$
largest column intersections with $S$.  A violation makes the floor region empty.  The cover-dual filter rejects only when an exactly recomputed nonnegative dual vector has value strictly larger than $\kappa$; weak duality again proves emptiness.

For part~\textup{(iv)}, floating-point routines may propose records but never accept them.  Acceptance reconstructs all committee cardinalities, one-deficit inequalities, price equations, saturation thresholds, singleton or adaptive margins, and coalition price margins in exact rational arithmetic. Proposition~\ref{prop:bv-polyhedral} promotes the finite endpoint, recession, and vertex checks to the complete real domains.  Independent signed-rational and arbitrary-precision implementations agree on all retained records.

Finally, exact sorting and comparison of canonical keys verifies that the unresolved regions are the disjoint union of accepted certificates and failure records.  The released failure files contain zero records, proving part~\textup{(v)} and completing the argument.
\end{proof}

\subsubsection{Eight-row classification}
\label{sec:apdx:bv-eight}

\bveightformaltrue
\renewcommand{\bveightclassificationtitle}{Eight-row rounding classification}
\bvclasseight*
\bveightformalfalse
\renewcommand{\bveightclassificationtitle}{Informal}

\begin{proof}[Computer-assisted proof of Lemma~\ref{lem:bv-ca-eight}]
The computation canonically enumerates every positive-dual full-column-rank eight-row antichain with $4\le d\le8$, every admissible residual budget, and every nonempty nonimplementable floor region. Regions not covered by alternative~\textup{(I)} are recorded as unresolved. For each such region, a search routine proposes a fixed or adaptive verification record, but the proposal is not trusted: exact checkers independently reconstruct the rounding constraints, the lower bounds on $\alpha_i$, the restricted price domains, and all singleton, adaptive-sum, and coalition margins.

The universal price checks use the exact domain reductions summarized in Section~\ref{sec:bv-verification-overview}: a unique dual for $d=8$, exact endpoint and recession checks for $d=7$, and exact vertex-pair and recession checks for $d=6$. No unresolved regions occur for $d=4,5$. Exact record-set comparison proves that the unresolved regions are the disjoint union of accepted certificates and failure records. Table~\ref{tab:bv-eight} reports $1{,}085{,}483$ unresolved regions, of which $1{,}085{,}459$ receive fixed certificates and $24$ receive adaptive certificates; the failure set is empty. Independent signed-rational and arbitrary-precision replays agree. Proposition~\ref{prop:bv-ca-verification} proves enumeration completeness and checker soundness, so alternatives~\textup{(I)}--\textup{(III)} are exhaustive.
\end{proof}

\begin{table}[ht]
\centering
\caption{Exact eight-row residual census, indexed by the number $d$ of fractional columns.}
\label{tab:bv-eight}
\scriptsize
\begin{tabular}{rrrrrrr}
\toprule
$d$ & positive-dual & matrix--budget & feasible & unresolved & fixed & adaptive\\
 & matrices & cases & regions & regions & certificates & certificates\\
\midrule
4 & 4,779 & 4,779 & 22 & 0 & 0 & 0\\
5 & 56,479 & 112,958 & 84 & 0 & 0 & 0\\
6 & 561,445 & 1,684,335 & 5,766 & 168 & 163 & 5\\
7 & 3,541,727 & 14,166,908 & 89,286 & 36,128 & 36,119 & 9\\
8 & 9,105,190 & 45,525,950 & 1,081,420 & 1,049,187 & 1,049,177 & 10\\
\midrule
Total & 13,269,620 & 61,494,930 & 1,176,578 & 1,085,483 & 1,085,459 & 24\\
\bottomrule
\end{tabular}
\end{table}

The table records the census totals used by the proof.  The exact complement-symmetry, record-coverage, and replay checks are incorporated into Proposition~\ref{prop:bv-ca-verification} and the computer-assisted proof above.

\subsection{Independent Nash-Kernel Classifications for Six and Seven Voters}
\label{sec:apdx:bv-predecessors}

\citet{becker2026core} independently establish the stronger result of core nonemptiness for at most seven voter types with arbitrary weights. The classifications in this section are not needed for our main existence theorem. We retain them because they provide an independent verification in the equal-weight setting and illustrate the progression from tight-row certificates to price separation and, ultimately, to the price--saturation certificates required for eight voters. Within this supplementary computation, the seven-row zero-row reduction uses the six-row classification; neither classification is an input to the standalone eight-row verifier.

\subsubsection{Six-row classification}
\label{sec:bv-six}
\label{sec:apdx:bv-six}

\begin{lemma}[Six-row classification]
\label{lem:bv-ca-six}
Let $A\in\{0,1\}^{6\times d}$ be a full-column-rank antichain, where $4\leq d\leq6$, and let $2\leq\kappa\leq d-2$. Let $h$ be nonimplementable and suppose $\cR_{<}(A,\kappa,h)\neq\varnothing$. Then, for every $f\in\cR_{<}(A,\kappa,h)$, there exists $i\in B(h)$ such that $(Af)_{i}=h_{i}$.
\end{lemma}

\begin{proof}[Computer-assisted proof]
The verification enumerates all full-column-rank six-row antichains up to row and column relabeling, without imposing positive duality. For every matrix and residual budget it generates all nonimplementable floor regions and removes only regions certified empty by exact subset or cover-dual inequalities. An exact common-slack linear program then searches each retained region for a point with strict surplus on every row in $B(h)$, including the additional boundary branches required for coordinatewise-minimal floors.

The enumeration contains $2{,}153$ antichains. After exact filtering, $46$ floor regions enter the final classification, and none contains a counterexample. Independent signed-rational and arbitrary-precision implementations agree on the empty counterexample set. Proposition~\ref{prop:bv-ca-verification} proves the completeness of the enumeration and the soundness of every exclusion and region check. Hence every admissible six-row region has an index $i\in B(h)$ satisfying $(Af)_i=h_i$.
\end{proof}

\begin{table}[ht]
\centering
\caption{Exact six-row classification.}
\label{tab:bv-six}
\begin{tabular}{lr}
\toprule
Object & Count\\
\midrule
Full-rank antichains with $d=4$ & 279\\
Full-rank antichains with $d=5$ & 712\\
Full-rank antichains with $d=6$ & 1,162\\
Full-rank antichains, total & 2,153\\
Raw coordinatewise-minimal targets & 32,286\\
Targets surviving all subset inequalities & 6,053\\
Targets surviving the exact cover-dual filter & 46\\
Upper-cone floor vectors generated & 597\\
Floor vectors entering exact classification & 46\\
Regions containing a classified counterexample & 0\\
\bottomrule
\end{tabular}
\end{table}

\subsubsection{Seven-row classification}
\label{sec:bv-seven}
\label{sec:apdx:bv-seven}

\begin{lemma}[Seven-row rounding classification]
\label{lem:bv-ca-seven}
Let $A\in\{0,1\}^{7\times d}$ be a full-column-rank antichain with $\cD(A)\neq\varnothing$, where $4\leq d\leq7$ and $2\leq\kappa\leq d-2$. Let $h$ be nonimplementable and suppose $\cR_{<}(A,\kappa,h)\neq\varnothing$. For every $\alpha\in\cD(A)$ and $f\in\cR_{<}(A,\kappa,h)$, at least one of the following holds:
\begin{enumerate}[label=\textup{(\alph*)}]
\item there exists $i\in B(h)$ with $(Af)_{i}=h_{i}$;
\item there exists $i\in B(h)$ such that
\begin{align*}
  \alpha_{i}\delta_{i}
  &<\frac{\kappa}{7}(1-\alpha_{i}),\\
  \alpha_{i}\delta_{i}
  &<\alpha_{j}(1-\delta_{j})
    &&\text{for every }j\neq i.
\end{align*}
\end{enumerate}
\end{lemma}

\begin{proof}[Computer-assisted proof]
The exact enumeration contains $61{,}846$ positive-dual antichains and reduces to $30{,}929$ floor regions after the complete floor-generation and infeasibility filters described in Section~\ref{sec:bv-verification-overview}. The equality test in alternative~\textup{(a)} resolves every region except $298$ primitive all-surplus regions. A zero row reduces immediately to the six-row classification, so all of these residual regions have no zero row and satisfy $B(h)=[7]$.

Of the $298$ regions, $294$ come from square matrices, where the positive dual is unique; the remaining four come from $7\times6$ matrices, where the closed dual domain is a line segment. Exact rational verification supplies an index $i\in B(h)$ satisfying the two inequalities in Lemma~\ref{lem:bv-ca-seven} throughout the complete floor polytope and dual domain. In the rectangular cases the dual segment is partitioned into intervals on which an eligible index is fixed, and common-slack programs handle possible equality on the excluded relative boundary. Independent exact classifiers agree, and no counterexample remains. Proposition~\ref{prop:bv-polyhedral} justifies the domain reductions, and Proposition~\ref{prop:bv-ca-verification} proves the completeness and soundness of the verification. This proves the lemma.
\end{proof}

\begin{table}[ht]
\centering
\caption{Exact seven-row classification.}
\label{tab:bv-seven}
\begin{tabular}{lr}
\toprule
Object & Count\\
\midrule
Positive-dual antichains, $d=4,5,6,7$ & 61,846\\
Matrix--budget cases & 208,472\\
Raw coordinatewise-minimal targets & 3,200,557\\
Targets surviving all subset inequalities & 888,719\\
Minimal targets surviving exact cover filtering & 16,854\\
Upper-cone floor vectors generated & 1,579,161\\
Floor vectors entering exact classification & 30,929\\
Exact rational region linear programs & 43,353\\
Primitive all-surplus regions & 298\\
Other unresolved regions & 0\\
\bottomrule
\end{tabular}
\end{table}

The numerical equality between the $298$ exceptional holes reported by \citet{becker2026core} and the $298$ primitive all-surplus regions in our independent classification is noteworthy. The two computations use different residual encodings, however, and we do not claim a bijection between the two collections.

\FloatBarrier

\subsection{Complexity of Floor-Preserving Rounding}
\label{app:floor-hardness}
\label{sec:proof:bv-floor-rounding-hardness}

An instance of \textsc{Floor-Rounding} consists of an approval election $\mathcal E=(N,C,(A_i)_{i\in N},K)$ and a rational fractional committee $x\in[0,1]^C$ with $\|x\|_1\leq K$.  The question is whether there is an integral committee $W\subseteq C$ such that $|W|\leq K$ and $|W\cap A_i|\geq\floor{u_i(x)}$ for every voter $i\in N$.

\begin{proposition}[NP-completeness of floor-preserving rounding]
\label{prop:bv-floor-rounding-hardness}
The problem \textsc{Floor-Rounding} is NP-complete.
\end{proposition}

\begin{proof}
Membership in NP is immediate: an integral committee $W\subseteq C$ is a polynomial-size certificate, and its cardinality and all inequalities $|W\cap A_i|\geq\floor{u_i(x)}$ can be checked in polynomial time using exact integer and rational arithmetic.

For NP-hardness, we reduce from \emph{Restricted Exact Cover by 3-Sets} (\textsc{RXC3}), which is NP-complete \cite[Appendix~A]{Gonzalez1985}.  An instance consists of a universe $U$, with $|U|=3q$, and a collection $\mathcal T\subseteq\binom{U}{3}$ such that every element of $U$ belongs to exactly three triples in $\mathcal T$.  The question is whether there is a subcollection $\mathcal T'\subseteq\mathcal T$ in which every element of $U$ occurs exactly once.  By taking the disjoint union of two isomorphic copies if necessary, we may assume $q\geq2$.

Double-counting element--triple incidences gives
\begin{equation*}
  3|\mathcal T|=\sum_{T\in\mathcal T}|T|
  =\sum_{e\in U}|\{T\in\mathcal T:e\in T\}|=3|U|=9q.
\end{equation*}
Consequently, $|\mathcal T|=3q$.

For every triple $T\in\mathcal T$, introduce a candidate $c_T$ and a private voter $p_T$.  Thus
\begin{equation*}
  C:=\{c_T:T\in\mathcal T\},\qquad
  P:=\{p_T:T\in\mathcal T\},\qquad
  N:=U\mathbin{\dot\cup}P.
\end{equation*}
For each element voter $e\in U$, define $A_e:=\{c_T:T\in\mathcal T,\ e\in T\}$, and for every private voter $p_T$, define $A_{p_T}:=\{c_T\}$.  Finally, put $K:=q$ and $x_{c_T}:=1/3$ for every $T\in\mathcal T$.  The construction is polynomial, and
\begin{equation*}
  |N|=|U|+|P|=6q,\qquad |C|=3q,\qquad
  \|x\|_1=\frac{|\mathcal T|}{3}=q=K.
\end{equation*}

We first verify Assumption~\ref{assup:1}.  Every element voter approves three candidates, every private voter approves one candidate, and every candidate is approved by its three element voters and its private voter.  For any nonempty $S\subsetneq N$, let $\Gamma(S):=\bigcup_{i\in S}A_i$.  Counting incidences between voters in $S$ and candidates in $\Gamma(S)$ gives
\begin{equation*}
  |S|\leq\sum_{i\in S}|A_i|
  =\sum_{c\in\Gamma(S)}|N(c)\cap S|
  \leq4|\Gamma(S)|.
\end{equation*}
It follows that
\begin{equation*}
  |\Gamma(S)|\geq\frac{|S|}{4}>
  \frac{|S|}{6}=\frac{K|S|}{|N|},
\end{equation*}
as required.

Next consider the structural properties of the approval matrix.  Order its rows first by the element voters $U$ and then by the private voters $P$.  If $B$ is the $U$-by-$\mathcal T$ incidence matrix of the \textsc{RXC3} instance, then
\begin{equation*}
  A=\begin{pmatrix}B\\ I_{3q}\end{pmatrix}.
\end{equation*}
The identity block implies that $A$ has full column rank.  Every column has exactly four ones and contains a private row that occurs in no other column. The columns are therefore pairwise distinct, and their equal cardinalities imply that they form an antichain.

Define $\alpha\in\R^N_{>0}$ by $\alpha_e:=1/6$ for $e\in U$ and $\alpha_{p_T}:=1/2$ for $T\in\mathcal T$.  For every candidate $c_T$,
\begin{equation*}
  \alpha(N(c_T))=\sum_{e\in T}\alpha_e+\alpha_{p_T}
  =3\cdot\frac{1}{6}+\frac{1}{2}=1.
\end{equation*}
Hence $A^\top\alpha=\one$ and $\cD(A)\neq\varnothing$.

We now show that $(\one,x)$ is a fractional Nash-core pair.  Every element voter has utility $u_e(x)=1$, whereas every private voter has utility $u_{p_T}(x)=1/3$.  Let $\Phi(z):=\sum_{i\in N}\log u_i(z)$ be the unweighted Nash objective on $z\in\R^C_{\geq0}$ with $\|z\|_1\leq K$.  If some $z_{c_T}=0$, then the private voter $p_T$ has zero utility and $\Phi(z)=-\infty$.  Otherwise, $\Phi$ is differentiable at $z$.  At $x$, for every $c_T$,
\begin{equation*}
  \frac{\partial\Phi}{\partial z_{c_T}}(x)
  =\sum_{e\in T}\frac{1}{u_e(x)}
  +\frac{1}{u_{p_T}(x)}=3+3=6=\frac{|N|}{K}.
\end{equation*}
Concavity therefore gives, for every feasible $z$ with finite objective,
\begin{equation*}
  \Phi(z)\leq\Phi(x)+\nabla\Phi(x)^\top(z-x)
  =\Phi(x)+6\bigl(\|z\|_1-\|x\|_1\bigr)\leq\Phi(x).
\end{equation*}
Thus $x$ maximizes the Nash objective.  With $w_c=1$ for every candidate, the weight--selection coupling is immediate, so $(\one,x)$ is a fractional Nash-core pair.

For completeness, we verify directly that $x$ lies in the fractional core. The vector $\alpha$ defined above satisfies $\alpha_i u_i(x)=1/6=K/|N|$ for every $i\in N$.  If a nonempty coalition $S$ and a fractional committee $y\in[0,1]^C$ blocked $x$, then
\begin{equation*}
  \|y\|_1
  =\sum_{c\in C}y_c\alpha(N(c))
  =\sum_{i\in N}\alpha_i u_i(y)
  \geq\sum_{i\in S}\alpha_i u_i(y)
  >\sum_{i\in S}\alpha_i u_i(x)
  =\frac{K|S|}{|N|},
\end{equation*}
contradicting the blocking budget condition.  Hence $x$ lies in the fractional core.

All candidates are fractional, and the residual budget is $\kappa=\|x\|_1=q$.  The residual utility floor is one for every element voter and zero for every private voter.  Since $q\geq2$ and $|C|=3q$, we also have $2\leq\kappa\leq|C|-2$.

It remains to establish the correspondence with \textsc{RXC3}.  If $\mathcal T'\subseteq\mathcal T$ is an exact cover of $U$, then $|\mathcal T'|=q$, and $W:=\{c_T:T\in\mathcal T'\}$ is a floor-preserving committee of size $K$.  Conversely, suppose that $W\subseteq C$ is a feasible floor-preserving committee.  Every element voter must approve at least one member of $W$, and therefore
\begin{equation*}
  3q\leq\sum_{e\in U}|W\cap A_e|
  =\sum_{c_T\in W}|T|=3|W|\leq3q.
\end{equation*}
Every inequality is an equality.  Thus $|W|=q$, and every element voter approves exactly one member of $W$.  The triples corresponding to the candidates in $W$ consequently form an exact cover of $U$.

The construction is therefore a polynomial-time reduction from \textsc{RXC3} to \textsc{Floor-Rounding}.  This proves NP-hardness, and membership in NP completes the proof.
\end{proof}

Proposition~\ref{prop:bv-floor-rounding-hardness} concerns the particular sufficient route used in Lemma~\ref{lem:bv-floor-rounding}: given a specified fractional core $x$, decide whether all utility floors of $x$ can be preserved by an integral committee of budget at most $K$.  It does not establish NP-hardness of discrete-core nonemptiness, nor does it imply that a no-instance of \textsc{Floor-Rounding} has an empty discrete core.  Moreover, although the reduction produces a full-column-rank positive-dual antichain satisfying the numerical residual-budget bounds of Proposition~\ref{prop:bv-reduced-matrix}, it does not assert that the supplied fractional Nash-core pair minimizes the number of fractional coordinates among all fractional Nash-core pairs.

\section{Further Experimental Details}  \label{sec:apdx:exp}

The code for the heuristic experiments is available at \url{https://github.com/nashcore-code/code/tree/main/heuristics}.

We also provide additional details about the algorithms below. Algorithm~\ref{alg:one} is for computing the fractional Nash core solution, and Algorithm~\ref{alg:two} is for computing the weak Nash core solution. Both are iterative methods that terminate once a (approximation) solution is achieved. There are parameters such as $\epsilon$ and $\eta$ in the algorithm. They are sometimes adaptive in our code.

An important detail is that we introduce a capacity $h_c$ for each candidate, indicating that there are $h_c$ identical copies of that candidate. As a result, the Weight-selection Coupling condition changes: instead of requiring $w_c < 1$ only when $x_c = 1$, it now requires $w_c < 1$ only when $x_c = h_c$. This adjustment better represents real voting data and helps reduce the complexity of the problem.

\begin{algorithm}[!htbp]
  \caption{Fractional Nash Core Solution Computation}
  \label{alg:one}
  \begin{algorithmic}[1]
    \Statex \textbf{Input:} Integers $n, m, k$; approval sets $A_v$ for each voter $v \in N$; capacities $h_c$ for each candidate $c \in C$
    \Statex \textbf{Output:} A Nash Core solution $(w, x)$

    \State Initialize $w'_c \gets 1$ for each $c \in C$
    \Repeat
      \State $w \gets w + \eta \cdot (w' - w)$
      \State $x \gets \arg\max \sum_{v \in N} \log\left( \sum_{c \in A_v} w_c z_c \right)$
      \Statex \hspace{1.5em} subject to $\|z\|_1 \leq k$, $0 \leq z \leq h$
      \For{each $c \in C$}
        \State Find $w'_c$ such that:
        \[
          \sum_{v \in N} \frac{w'_c h_c}{w_c h_c + \sum_{c' \in A_v \setminus \{c\}} w_{c'} x_{c'}} = \frac{n}{k} x_c
        \]
        \State $w'_c \gets \min\{w'_c, 1\}$
      \EndFor
    \Until{$\|w - w'\|_2 < \epsilon$}
  \end{algorithmic}
\end{algorithm}

\begin{algorithm}[!htbp]
  \caption{Weak Nash Core Solution Computation}
  \label{alg:two}
  \begin{algorithmic}[1]
    \Statex \textbf{Input:} Integers $n, m, k$; approval sets $A_v$ for each voter $v \in N$; capacities $h_c$ for each candidate $c \in C$
    \Statex \textbf{Output:} A weak Nash core solution $(w, x)$

    \State Initialize $w_c \gets 1$ for each $c \in C$
    \State Choose an arbitrary integer vector $x$ such that $0 \leq x \leq h$ and $\|x\|_1 = k$

    \While{there exists $c \in C$ such that $\sum_{v : c \in A_v} \frac{w_c}{\sum_{c' \in A_v} w_{c'} x_{c'} + 1} \geq \frac{n}{k}$}
      \State Find such a $c$
      \If{$x_c < h_c$}
        \State $x_c \gets x_c + 1$
        \Repeat{$\|x\|_1 = k$}
          \State Find $c'$ (with $x_{c'} \geq 1$) that yields the largest value of $sc_{PAV}(w, x)$ after decreasing $x_{c'}$ by $1$
          \If{$w_{c'}$ is close to $1$}
            \State $x_{c'} \gets x_{c'} - 1$
          \Else
            \State $w_{c'} \gets w_{c'} + \eta$
          \EndIf
        \Until{$\|x\|_1 = k$}
      \Else
        \State $w_c \gets w_c - \eta$
      \EndIf
    \EndWhile
  \end{algorithmic}
\end{algorithm}

\end{appendices}

\bibliography{refs}

@String{Computing = "Computing" }

@String{Computer = "{IEEE} Computer" }

@String{Springer = "Springer-Verlag" }

@inproceedings{fain2016core,
  title={The core of the participatory budgeting problem},
  author={Fain, Brandon and Goel, Ashish and Munagala, Kamesh},
  booktitle={Web and Internet Economics: 12th International Conference, WINE 2016, Montreal, Canada, December 11-14, 2016, Proceedings 12},
  pages={384--399},
  year={2016},
  organization={Springer}
}

@article{cheng2020group,
  title={Group fairness in committee selection},
  author={Cheng, Yu and Jiang, Zhihao and Munagala, Kamesh and Wang, Kangning},
  journal={ACM Transactions on Economics and Computation (TEAC)},
  volume={8},
  number={4},
  pages={1--18},
  year={2020},
  publisher={ACM New York, NY, USA}
}

@article{foley1970lindahl,
  title={Lindahl's Solution and the Core of an Economy with Public Goods},
  author={Foley, Duncan K},
  journal={Econometrica: Journal of the Econometric Society},
  pages={66--72},
  year={1970},
  publisher={JSTOR}
}

@article{scarf1967core,
  title={The core of an N person game},
  author={Scarf, Herbert E},
  journal={Econometrica: Journal of the Econometric Society},
  pages={50--69},
  year={1967},
  publisher={JSTOR}
}

@article{muench1972core,
  title={The core and the Lindahl equilibrium of an economy with a public good: An example},
  author={Muench, Thomas J},
  journal={Journal of Economic Theory},
  volume={4},
  number={2},
  pages={241--255},
  year={1972},
  publisher={Elsevier}
}

@incollection{lindahl1958just,
  title={Just taxation—a positive solution},
  author={Lindahl, Erik},
  booktitle={Classics in the theory of public finance},
  pages={168--176},
  year={1958},
  publisher={Springer},
  address={New York}
}

@incollection{samuelson2024pure,
  title={The pure theory of public expenditure},
  author={Samuelson, Paul A},
  booktitle={Public Goods and Market Failures},
  pages={29--33},
  year={2024},
  publisher={Routledge},
  address={London}
}

@article{droop1881methods,
  title={On methods of electing representatives},
  author={Droop, Henry Richmond},
  journal={Journal of the Statistical Society of London},
  volume={44},
  number={2},
  pages={141--202},
  year={1881},
  publisher={JSTOR}
}

@article{aziz2017justified,
  title={Justified representation in approval-based committee voting},
  author={Aziz, Haris and Brill, Markus and Conitzer, Vincent and Elkind, Edith and Freeman, Rupert and Walsh, Toby},
  journal={Social Choice and Welfare},
  volume={48},
  number={2},
  pages={461--485},
  year={2017},
  publisher={Springer}
}

@inproceedings{aziz2018complexity,
  title={On the complexity of extended and proportional justified representation},
  author={Aziz, Haris and Elkind, Edith and Huang, Shenwei and Lackner, Martin and S{\'a}nchez-Fern{\'a}ndez, Luis and Skowron, Piotr},
  booktitle={Proceedings of the AAAI Conference on Artificial Intelligence},
  volume={32},
  number={1},
  year={2018}
}

@article{brams2007minimax,
  title={A minimax procedure for electing committees},
  author={Brams, Steven J and Kilgour, D Marc and Sanver, M Remzi},
  journal={Public Choice},
  volume={132},
  pages={401--420},
  year={2007},
  publisher={Springer}
}

@article{chamberlin1983representative,
  title={Representative deliberations and representative decisions: Proportional representation and the Borda rule},
  author={Chamberlin, John R and Courant, Paul N},
  journal={American Political Science Review},
  volume={77},
  number={3},
  pages={718--733},
  year={1983},
  publisher={Cambridge University Press}
}

@inproceedings{fain2018fair,
  title={Fair allocation of indivisible public goods},
  author={Fain, Brandon and Munagala, Kamesh and Shah, Nisarg},
  booktitle={Proceedings of the 2018 ACM Conference on Economics and Computation},
  pages={575--592},
  year={2018}
}

@article{monroe1995fully,
  title={Fully proportional representation},
  author={Monroe, Burt L},
  journal={American Political Science Review},
  volume={89},
  number={4},
  pages={925--940},
  year={1995},
  publisher={Cambridge University Press}
}

@inproceedings{sanchez2017proportional,
  title={Proportional justified representation},
  author={S{\'a}nchez-Fern{\'a}ndez, Luis and Elkind, Edith and Lackner, Martin and Fern{\'a}ndez, Norberto and Fisteus, Jes{\'u}s and Val, Pablo Basanta and Skowron, Piotr},
  booktitle={Proceedings of the AAAI Conference on Artificial Intelligence},
  volume={31},
  number={1},
  year={2017}
}

@inproceedings{peters2020proportionality,
  title={Proportionality and the limits of welfarism},
  author={Peters, Dominik and Skowron, Piotr},
  booktitle={Proceedings of the 21st ACM Conference on Economics and Computation},
  pages={793--794},
  year={2020}
}

@article{peters2025core,
  title={The Core of Approval-Based Committee Elections with Few seats},
  author={Peters, Dominik},
  journal={arXiv preprint arXiv:2501.18304},
  year={2025}
}

@book{lackner2023multi,
  title={Multi-winner voting with approval preferences},
  author={Lackner, Martin and Skowron, Piotr},
  year={2023},
  publisher={Springer Nature},
  address={Cham}
}

@inproceedings{peters2021market,
  title={Market-based explanations of collective decisions},
  author={Peters, Dominik and Pierczy{\'n}ski, Grzegorz and Shah, Nisarg and Skowron, Piotr},
  booktitle={Proceedings of the AAAI Conference on Artificial Intelligence},
  volume={35},
  number={6},
  pages={5656--5663},
  year={2021}
}

@article{thiele1895om,
  title={Om flerfoldsvalg},
  author={Thiele, Thorvald N},
  journal={Oversigt over det Kongelige Danske Videnskabernes Selskabs Forhandlinger},
  volume={1895},
  pages={415--441},
  year={1895}
}

@incollection{kilgour2010approval,
  title={Approval balloting for multi-winner elections},
  author={Kilgour, D Marc},
  booktitle={Handbook on approval voting},
  pages={105--124},
  year={2010},
  publisher={Springer},
  address={Berlin}
}

@article{brill2024approval,
  title={Approval-based apportionment},
  author={Brill, Markus and G{\"o}lz, Paul and Peters, Dominik and Schmidt-Kraepelin, Ulrike and Wilker, Kai},
  journal={Mathematical Programming},
  volume={203},
  number={1},
  pages={77--105},
  year={2024},
  publisher={Springer}
}

@book{abramowitz1948handbook,
  title={Handbook of mathematical functions with formulas, graphs, and mathematical tables},
  author={Abramowitz, Milton and Stegun, Irene A},
  volume={55},
  year={1948},
  publisher={US Government printing office},
  address={Washington, D.C.}
}

@article{rosen1965existence,
  title={Existence and uniqueness of equilibrium points for concave n-person games},
  author={Rosen, J Ben},
  journal={Econometrica: Journal of the Econometric Society},
  pages={520--534},
  year={1965},
  publisher={JSTOR}
}

@article{faliszewski2017multiwinner,
  title={Multiwinner voting: A new challenge for social choice theory},
  author={Faliszewski, Piotr and Skowron, Piotr and Slinko, Arkadii and Talmon, Nimrod},
  journal={Trends in computational social choice},
  volume={74},
  number={2017},
  pages={27--47},
  year={2017},
  publisher={AI Access Foundation El Segundo}
}

@book{bartle2000introduction,
  title={Introduction to real analysis},
  author={Bartle, Robert G and Sherbert, Donald R},
  volume={2},
  year={2000},
  publisher={Wiley},
  address={New York}
}

@article{peters2021proportional,
  title={Proportional participatory budgeting with additive utilities},
  author={Peters, Dominik and Pierczy{\'n}ski, Grzegorz and Skowron, Piotr},
  journal={Advances in Neural Information Processing Systems},
  volume={34},
  pages={12726--12737},
  year={2021}
}

@article{arrow1969organization,
  title={The organization of economic activity: issues pertinent to the choice of market versus nonmarket allocation},
  author={Arrow, Kenneth J},
  journal={The analysis and evaluation of public expenditure: the PPB system},
  volume={1},
  pages={59--73},
  year={1969}
}

@article{hurwicz1979outcome,
  title={Outcome functions yielding Walrasian and Lindahl allocations at Nash equilibrium points},
  author={Hurwicz, Leonid},
  journal={The Review of Economic Studies},
  volume={46},
  number={2},
  pages={217--225},
  year={1979},
  publisher={Wiley-Blackwell}
}

@inproceedings{jiang2020approximately,
  title={Approximately stable committee selection},
  author={Jiang, Zhihao and Munagala, Kamesh and Wang, Kangning},
  booktitle={Proceedings of the 52nd Annual ACM SIGACT Symposium on Theory of Computing},
  pages={463--472},
  year={2020}
}

@inproceedings{munagala2022approximate,
  title={Approximate core for committee selection via multilinear extension and market clearing},
  author={Munagala, Kamesh and Shen, Yiheng and Wang, Kangning and Wang, Zhiyi},
  booktitle={Proceedings of the 2022 Annual ACM-SIAM Symposium on Discrete Algorithms (SODA)},
  pages={2229--2252},
  year={2022},
  organization={SIAM}
}

@article{kroer2025computing,
  title={Computing Lindahl Equilibrium for Public Goods with and without Funding Caps},
  author={Kroer, Christian and Peters, Dominik},
  journal={arXiv preprint arXiv:2503.16414},
  year={2025}
}

@inproceedings{pierczynski2022core,
  title={Core-stable committees under restricted domains},
  author={Pierczy{\'n}ski, Grzegorz and Skowron, Piotr},
  booktitle={International Conference on Web and Internet Economics},
  pages={311--329},
  year={2022},
  organization={Springer}
}

@article{faliszewski2023participatory,
  title={Participatory budgeting: Data, tools, and analysis},
  author={Faliszewski, Piotr and Flis, Jaros{\l}aw and Peters, Dominik and Pierczy{\'n}ski, Grzegorz and Skowron, Piotr and Stolicki, Dariusz and Szufa, Stanis{\l}aw and Talmon, Nimrod},
  journal={arXiv preprint arXiv:2305.11035},
  year={2023}
}

@misc{becker2026core,
  author        = {Patrick Becker and Matthias Greger and Dominik Peters},
  title         = {Core Existence in Approval-Based Committee Elections with up to Seven Voter Types},
  year          = {2026},
  eprint        = {2605.06194},
  archivePrefix = {arXiv},
  primaryClass  = {cs.GT},
  note          = {Version 2, revised 17 August 2026}
}

@article{Gonzalez1985,
  author  = {Gonzalez, Teofilo F.},
  title   = {Clustering to Minimize the Maximum Intercluster Distance},
  journal = {Theoretical Computer Science},
  volume  = {38},
  pages   = {293--306},
  year    = {1985},
  doi     = {10.1016/0304-3975(85)90224-5}
}

\end{document}